\documentclass[a4paper,UKenglish,thm-restate,cleveref,numberwithinsect]{lipics-v2021}

\newboolean{arXivLongVersion}
\setboolean{arXivLongVersion}{true} %

\pdfoutput=1 %

\ifthenelse{\boolean{arXivLongVersion}}{%
\hideLIPIcs  %
}{%
\relatedversiondetails{Full Version}{https://arxiv.org/abs/XXX}
}

\EventEditors{Ana Sokolova and Patrick Totzke}
\EventNoEds{2}
\EventLongTitle{37th International Conference on Concurrency Theory (CONCUR 2026)}
\EventShortTitle{CONCUR 2026}
\EventAcronym{CONCUR}
\EventYear{2026}
\EventDate{September 1--4, 2026}
\EventLocation{Liverpool, UK}
\EventLogo{}
\SeriesVolume{391}
\ArticleNo{20}

\nolinenumbers

\newcommand{\rendu}{}

\usepackage{amsmath,amsfonts}
\usepackage{stmaryrd}
\usepackage[x11names]{xcolor}
\usepackage{my_msc}

\usepackage{subfiles}
\usepackage[textsize=footnotesize, textwidth=25mm]{todonotes}
\usepackage{logic9-thm}
\usepackage{quotes}
\usepackage{tabularx,booktabs}
\usepackage{subcaption}
\usepackage[hyperref,quotation,makeidx]{knowledge}
\usepackage{hyperref}

\hypersetup{
    colorlinks=true,
    linkcolor=blue,  
    urlcolor=cyan,
    citecolor=Green4
    }

\usepackage{marvosym}
\usepackage{makecell,multirow}
\usepackage[alpine]{ifsym}
\usepackage{changepage}
\usepackage{multicol}

\usepackage{tikz}
\usetikzlibrary{
  arrows.meta,bending,positioning,backgrounds,calc,
  automata,arrows,decorations.pathreplacing,calligraphy,
  shapes.symbols
}
\tikzset{>=stealth'} %

\tikzstyle{every picture} = [style=semithick]
\tikzstyle{every node}    = [font=\small]
\tikzstyle{every state}   = [thick, minimum size=5mm, inner sep=0pt]

\tikzstyle{initial}   = [initial   by arrow, initial   text=, initial   distance=4mm]
\tikzstyle{accepting} = [accepting by arrow, accepting text=, accepting distance=4mm]

\usepackage{cancel}

\newcommand{\Ch}{\mathit{Ch}}

\newcommand{\act}[1]{\xrightarrow{#1}}

\newcommand{\Msg}{\mathbb{M}}

\makeatletter
\newcommand{\dashbar}[2][\mathop]{#1{\mathpalette\df@over{{\dashfill}{#2}}}}
\newcommand{\df@over}[2]{\df@@over#1#2}
\newcommand\df@@over[3]{%
  \vbox{
    \offinterlineskip
    \ialign{##\cr
      #2{#1}\cr
      \noalign{\kern1pt}
      $\m@th#1#3$\cr
    }
  }%
}
\newcommand{\dashfill}[1]{%
  \kern-.5pt
  \xleaders\hbox{\kern.5pt\vrule height.6pt width \dash@width{#1}\kern.5pt}\hfill
  \kern-.5pt
}
\newcommand{\dash@width}[1]{%
  \ifx#1\displaystyle
    2pt
  \else
    \ifx#1\textstyle
      1.5pt
    \else
      \ifx#1\scriptstyle
        1.25pt
      \else
        \ifx#1\scriptscriptstyle
          1pt
        \fi
      \fi
    \fi
  \fi
}
\makeatother

\ifx\mscset\undefined
\else
  \mscset{
      /msc/draw head/.initial=none,       %
      /msc/draw head/.default=none,
      /msc/draw frame/.initial=none, 
      /msc/draw frame/.default=none,
      instance end/.initial=foot,
      /msc/left inline overlap=7.5pt,
      /msc/right inline overlap=7.5pt,
      stop width=0.2cm,
      message/.style={arrows=Circle->. Circle,shorten > = -2.5pt,shorten < = -2.5pt,thick},
      /msc/region width = 0pt,
      /msc/foot height = 0pt,
      /msc/foot distance = 2pt,
      /msc/left environment distance = 0pt,
      /msc/right environment distance = 0pt,
      /msc/head top distance = 0pt,
      /msc/title distance = 0pt,
      /msc/title top distance = 0pt
      }
  \tikzset{>=stealth}

\fi
\ifx\algorithmiccomment\undefined
\else
  \renewcommand{\algorithmiccomment}[1]{{\color{gray}// #1}}
\fi

\IfClassLoadedTF{beamer}{}
{%
\ifx\lemma\undefined
  \newtheorem{lemma}{Lemma}[section]
  \newtheorem{lemma*}{Lemma}[section]
\else
\fi

\ifx\corollary\undefined
  \newtheorem{corollary}[lemma]{Corollary}
\else
\fi

\ifx\theorem\undefined
  \newtheorem{theorem}[lemma]{Theorem}
  \newtheorem{theorem*}[lemma]{Theorem}
\else
\fi

\ifx\definition\undefined
  \theoremstyle{break}
  \theorembodyfont{ \normalfont}
  \newtheorem{definition}[lemma]{Definition}
  \newtheorem{definition*}[lemma]{Definition}
\else

\fi

\theoremstyle{remark}

\ifx\remark\undefined
  \newtheorem{remark}[lemma]{Remark}
  \newtheorem{remark*}[lemma]{Remark}
\else

  \fi

  \ifx\example\undefined
  \newtheorem{example}{Example}
  \newtheorem{example*}{Example}
\else

\fi

\ifx\notation\undefined
\newenvironment{notation}[1]{
  \vspace{5pt}
  \setlength{\parindent}{-10pt}%
  \setlength{\leftskip}{10pt}%
  \setlength{\rightskip}{0pt}%
    \textbf{#1} :}
{\par\vspace{5pt}}
\fi
\ifx\proof\undefined
  \newenvironment{proof}[1]{
  \noindent
  \textbf{proof: }
 \begin{adjustwidth}{0.25cm}{0.25cm}}
  {\hfill$\Box$
 \end{adjustwidth}}
\fi

\DeclareDocumentEnvironment{mydefinition}{O{} O{ }}{%
  \begin{definition}[{#1}]%
  \ifx#2 
  \else
    \label{def:#2}%
  \fi
}{\end{definition}}
\DeclareDocumentEnvironment{mylemma}{O{ }}
{\restatable{lemma}{lem:#1}%
\ifx#1 
\else
  \label{lem:#1}%
\fi}
{\endrestatable}
\DeclareDocumentEnvironment{mycorollary}{O{ }}
{\restatable{corollary}{cor:#1}%
  \ifx#1 
  \else
    \label{cor:#1}%
  \fi}
  {\endrestatable}
\DeclareDocumentEnvironment{mytheorem}{O{ }}
{\restatable{theorem}{thm:#1}%
\ifx#1 
\else
  \label{thm:#1}%
\fi}
{\endrestatable}
\DeclareDocumentEnvironment{myremark}{O{ }}
{\begin{remark}%
  \ifx#1 
  \else
    \label{rem:#1}%
  \fi}
  {\end{remark}}

\newenvironment{proof-sketch}{\begin{proof}{\hspace{-5pt}\textbf{\color{gray}[sketch]}}}{\end{proof}}
}
\makeatletter
\newcommand{\myrecall}[1]{\csname#1\endcsname*}
\makeatother
\definecolor{todocol}{RGB}{249,226,52}

\definecolor{commentColRomain}{RGB}{29, 182, 82}
\definecolor{commentColAnca}{RGB}{182, 29, 82}
\definecolor{commentColGregoire}{RGB}{148, 0 211}
\newcommand{\pcomment}[3]{{\small\color{#1}#3\{#2\}}}
\ifx\rendu\undefined
\newcommand{\anca}[1]{\pcomment{commentColAnca}{#1}{A}}

\newcommand{\romain}[1]{{\pcomment{commentColRomain}{#1}{R}}}
\newcommand{\gregoire}[1]{{\pcomment{commentColGregoire}{#1}{G}}}
\else

\newcommand{\anca}[1]{\undefined}

\newcommand{\romain}[1]{{\undefined}}
\newcommand{\gregoire}[1]{{undefined}}
\fi

\usepackage{quotes}
\knowledgeconfigure{diagnose line=true}
\knowledgedirective{notion}{autoref,style=standard,intro style=intro standard}
\knowledgeconfigure{quotation}
\knowledgestyle{standard}{color=Blue4}
\knowledgestyle{intro standard}{emphasize} 

\knowledge{notion}
 | Parametrized Communicating Automaton
 | PCA

\knowledge{notion}
 | $k$-phase-bounded

\knowledge{notion}
 | $(d,k)$-coverable
 | $(\infty ,1)$-coverable
 | $(1,1)$-coverable
 | $(1,\infty )$-coverable
 | $(d,k=\infty )$-coverable
 | $({d=\infty },{k})$-coverable

\knowledge{notion}
 | VASS

\knowledge{notion}
 | Nested Counter System
 | NCS
 | Nested Counter Systems

\knowledge{notion}
 | Nested Reset Counter System
 | NRCS
 | Nested Reset Counter Systems

\knowledge{notion}
 | transducer-defined Petri Net
 | TdPN

\knowledge{notion}
 | Interface PCA
 | IPCA
\knowledge{notion}
 | VASS with actions

\knowledge{notion}
 | interface language

\def\ie{{i.e.}}

\title{On Parameterized Verification Over Tree Topologies}

\author{Romain Delpy}{LaBRI, Univ. Bordeaux, CNRS, Bordeaux INP, Talence, France}{}{https://orcid.org/0009-0006-0716-3787}{}
\author{Anca Muscholl}{LaBRI, Univ. Bordeaux, CNRS, Bordeaux INP, Talence, France}{}{https://orcid.org/0000-0002-8214-204X}{}
\author{Grégoire Sutre}{LaBRI, Univ. Bordeaux, CNRS, Bordeaux INP, Talence, France}{}{https://orcid.org/0009-0004-3839-0005}{}

\authorrunning{R. Delpy, A. Muscholl, and G. Sutre}

\Copyright{Romain Delpy, Anca Muscholl, and Grégoire Sutre}



\funding{%
  This work was (partially) supported by the grant ANR-23-CE48-0005 of the French National
  Research Agency ANR (project PaVeDyS).
}

\newcommand{\projonproc}[2]{{#1}_{|#2}}

\newcommand{\PSnd}{\mathtt{{Dw}}}
\newcommand{\PRec}{\mathtt{{Up}}}

\newcommand{\Stt}{S}
\newcommand{\init}{\text{\textit{init}}}
\newcommand{\final}{\text{\textit{final}}}
\newcommand{\full}{\text{\textit{full}}}
\newcommand{\dead}{\text{\textit{dead}}}
\newcommand{\idle}{\text{\textit{idle}}}
\newcommand{\Tmove}{\Mm_\text{\textit{move}}}
\newcommand{\Tjoin}{\Mm_\text{\textit{join}}}
\newcommand{\Tfork}{\Mm_\text{\textit{fork}}}
\newcommand{\winit}{w_\text{\textit{init}}}
\newcommand{\wfinal}{w_\text{\textit{final}}}
\newcommand{\mrk}{\mathbf{m}}
\newcommand{\start}{\mathsf{start}}
\newcommand{\Conf}{\mathcal{C}}

\renewcommand{\root}{\mathsf{root}}

\newcommand{\cnf}{\mathbf{c}}
\renewcommand{\mrk}{\mathbf{m}}
\newcommand{\initLoc}{s^{init}}
\newcommand{\cnfi}[2]{\mathbf{c}^\init_{#1,#2}}
\newcommand{\sycom}[3]{#1 {\Join} #2 (#3)}
\newcommand{\sytau}[1]{#1 : \tau}
\newcommand{\Reachdk}[2]{\textsc{$(#1,#2)$-Coverability}}

\newcommand{\ivec}{\mathbf{x}}
\newcommand{\ivtr}{\mathbf{v}}
\newcommand{\vzero}{\mathbf{0}}
\newcommand{\trace}{\mathsf{trace}}
\renewcommand{\Msg}{M}
\newcommand{\Sigmat}{A_\tau}
\newcommand{\transPCA}{\Delta}
\newcommand{\VtransPCA}{\Delta_A}
\newcommand{\transVASS}{\Delta}
\newcommand{\VtransVASS}{\Delta_V}
\newcommand{\Natinf}{\Nat_\infty}
\newcommand{\ncnf}{C}
\newcommand{\FF}{\mathbf{F}}
\newcommand{\up}{{\uparrow}}
\newcommand{\dw}{{\downarrow}}
\newcommand{\Tr}{\Tt}
\newcommand{\upd}{\mathrm{upd}}
\newcommand{\rst}[1]{\mathrm{rst}(p)}

\begin{document}

\maketitle

\begin{abstract}
  Parameterized verification of finite-state processes with
  rendez-vous synchronization is notoriously undecidable when processes are
  linearly ordered. In this paper we study two kinds of bounds under which we
  determine the complexity  of safety checking over tree topologies. When
  bounding the depth we obtain that the complexity is related to the fast
  growing hierarchy. Our second  bound limits the alternations between upwards
  and downwards synchronizations in the tree (phases), and occurs naturally in
  many concrete settings.  If we fix the number of phases then the complexity of
  safety checking is \EXPSPACE\ complete, and if the number of phases is part of
 the input it is 2\EXPSPACE\ complete (both for arbitrary depth).
\end{abstract}

\section{Introduction}

Analyzing the correctness of concurrent or distributed systems is an inherently
challenging task, due to inter-process interactions that can result in very
large state spaces.
Beyond this state explosion problem, a different challenge arises in
settings where one is interested in proving
correctness independently of the system size. 
This is the realm of parameterized verification, that has gained intensive
momentum over the past decade (cf.~surveys~\cite{Veith15,Esparza14}). 

In this paper we investigate the complexity of safety checking of  parameterized
systems over tree topologies. Our focus on trees is motivated by two key ideas.
First, trees naturally model modern execution environments: in recursive process
creation, parent and child processes synchronize through shared communication
channels. This structure also arises in cloud infrastructures (master-worker
hierarchies), file systems, and hierarchical cache coherence protocols. Second,
while parameterized communication over general graphs is notoriously difficult
to define — often requiring complex machinery like MSO and bounded
tree-width~\cite{AminofKRSV18}, or specialized logics~\cite{BozgaEISW20,
BozgaI21} — tree topologies offer a balance between expressive power
and algorithmic tractability.

Consider a concurrent web scraper, modeled according to a recursive \texttt{Rust} implementation
(\href{https://github.com/dawksh/mscraper/}{\color{blue}\texttt{github.com/dawksh/mscraper/}})
which scans for URLs embedded in web pages.
In this system, when a thread discovers a new URL, it spawns a child thread to
scan the corresponding page; this process continues recursively until a
predefined depth is reached. Communication follows the tree structure: each
process aggregates URL lists from its children and propagates the results back
to its parent. Formally, we model each process as a finite-state automaton where
transitions represent either synchronizations with a parent (written as
$\up$), or with a child (written as $\dw$, see Figure~\ref{fig:example}). 
Data exchange involves passing URLs coupled with depth counters; once a process receives a depth value equal to the maximum threshold, it terminates its recursive spawning and stops synchronization with further child processes.

\begin{figure}
    \begin{center}
        \begin{subfigure}{0.3\textwidth}
            \begin{tikzpicture}[initial text=,
    every state/.style={inner sep=2pt,minimum size=10pt,ellipse}]

    \node[state,initial,circle] (init) at (0,0) {};
    \node[state] (end) at (2,0) {};
    \node        () at (0,-1) {\color{white}a};

    \path[->]   (init) edge node[above] {$\dw(\texttt{url},1)$} (end)
                (init) edge[loop above] node[above] {$\dw(\texttt{url},1)$} ()
                (end) edge[loop above] node[above] {$\dw$List<\texttt{url}>} ()
                ;

\end{tikzpicture}
        \end{subfigure}
        \hfill
        \begin{subfigure}{0.6\textwidth}
        \begin{tikzpicture}[initial text=,
    every state/.style={inner sep=2pt,minimum size=10pt,ellipse}]

    \node[state,initial,circle] (init) at (0,0) {};
    \node[state] (dup1) at (2.1,1) {};
    \node[state] (dup3) at (6.3,1) {};
    \node[state] (dup2) at (4.2,1) {};
    \node[state] (bot1) at (2.1,-1) {};
    \node[state] (bot2) at (4.2,-1) {};
    \node        ()     at (0,1.5) {for $i<d$};

    \path[->]   (init) edge node[above,sloped] {$\up(\texttt{url},i)$} (dup1)
                (dup1) edge[loop above] node[above] {$\dw(\texttt{url},i+1)$} ()
                (dup1) edge node[above] {$\dw(\texttt{url},i+1)$} (dup2)
                (dup2) edge[loop above] node[above] {$\dw$List<\texttt{url}>} ()
                (dup2) edge node[above] {$\up$List<\texttt{url}>} (dup3)
                (init) edge node[above,sloped] {$\up(\texttt{url},d)$} (bot1)
                (bot1) edge node[above] {$\up$List<\texttt{url}>} (bot2)
                ;

\end{tikzpicture}
        \end{subfigure}
    \end{center}
    \caption{Automata representing the processes of the web scrapper, where $d$ is the maximal depth.
    The left one represents the main process, and the right one represents the spawned threads.}
    \label{fig:example}
\end{figure}
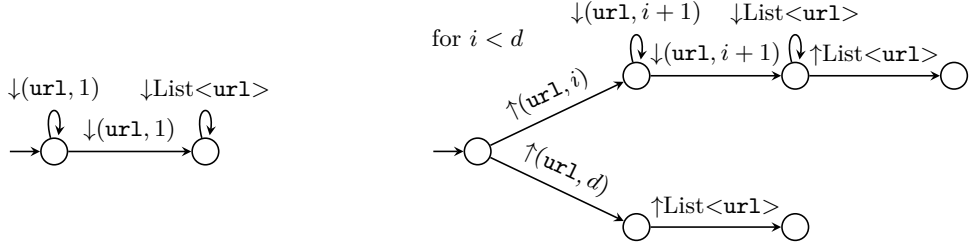

In  models with peer-to-peer communication, the undecidability of
safety  checking is immediate under asynchronous communication~\cite{brand1983communicating}, regardless of parameterization.
We therefore focus on  rendez-vous
synchronization. Even in this synchronous setting, however, safety checking on
trees, or even arrays, remains  undecidable~\cite{AptK86,german1992reasoning}. 
This motivates the search for restrictions under which safety checking is decidable.

We investigate the complexity of two under-approximations of systems with rendez-vous
synchronization over trees. The first
one consists in bounding the tree depth. This is arguably the most
natural restriction that we can impose on our systems. However, we show that
in this case the complexity of safety checking is prohibitively high,
characterized by the fast growing hierarchy. Specifically, for a fixed depth
$d$, complexity lies between the levels $\FF_{\Omega_{d-1}}$ and
$\FF_{\Omega_{d}}$ of the hierarchy.
This result motivates our second restriction, called \emph{phase
bound}. Within a phase a process may synchronize either exclusively with its parent or exclusively with its children. Our web scraper
example above has phase bound 3. It is also interesting to note that many of the
standard examples of parametrized protocols over trees are phase bounded,
e.g.~the Leader election protocol~\cite{AbdullaHDHR08} (electing a leader among leaves) has also
phase bound 3.  We show that by fixing the number of
phases, safety checking becomes \EXPSPACE-complete, and if the number of phases
is part of
the input then it becomes
2\EXPSPACE-complete. Somewhat surprisingly, these complexity results hold
independently of the tree depth.

\emph{Related work.} 
The verification of parameterized tree-structured systems has been extensively
studied. Tree Regular Model Checking 
\cite{AbdullaJMd02, KestenMMPS01} provides a  symbolic framework
relying on tree transducers. It often results in semi-algorithmic
procedures and is typically restricted to ranked architectures. To cope with
non-termination and implementation complexity, alternative symbolic approaches
have been proposed, such as backward reachability~\cite{AbdullaHDHR08} and
abstraction-refinement techniques~\cite{BouajjaniHRV12}.

A parallel line of research focuses on cutoff-based verification, which reduces
the parameterized problem to a finite check on a representative system size.
While cutoffs have been established for various logics, including indexed LTL
and CTL \cite{AminofKRSV18}, and their existence proved to be polynomially
decidable for certain rendez-vous protocols \cite{BalasubramanianER23}, this
approach remains sensitive to the state-space explosion problem.

Our work also relates to recent efforts to identify decidable fragments through
structural and behavioral constraints. In the context of pi-calculus,
bounding the tree depth of process creation leads to a complexity of
$\FF_{\varepsilon_0}$~\cite{balasubramanian2022complexity}. A different phase-based abstraction has been explored for broadcast
communication \cite{GuillouSS24}. In that setting, a phase is defined by the
direction of communication (exclusively broadcasting or receiving); however,
safety remains undecidable even for a small number of phases.

\smallskip
For convenience, technical terms and notations in the
electronic version of this manuscript are hyper-linked to their 
definitions (cf.~\url{https://ctan.org/pkg/knowledge}).
\ifthenelse{\boolean{arXivLongVersion}}{%
Proofs that are missing in the main text can be found in the
appendix.
}{}

\section{Definitions}

In this paper, we consider communicating processes over tree topologies.
A rooted tree  $\Tr=(P,\Ch)$ is a directed graph without self-loops, where
$P$ is a set of processes and $\Ch \subseteq P\times P$ is a set of
channels, with one distinct node (the root) such that every
node has a unique path from the root.
We will denote the root of tree $\Tr$  by $\root_\Tr$ (or
just $\root$ if $\Tr$ is clear from the context).
The orientation of the edges specifies how processes interact in the tree: $(p,q)\in \Ch$ means process $p$ is the parent of process $q$.

We let %
$\Msg$ denote a finite set of message contents, and we  write:
\begin{itemize}
  \item $\PSnd = \{\dw m\mid m\in\Msg\}$ 
  for the set of synchronization actions with some child process (downwards synchronizations).
  \item $\PRec = \{\up m\mid  m\in \Msg\}$ for the set of synchronization
  actions with the parent process (upwards synchronizations).
\end{itemize}
By $\Act = \PSnd \cup \PRec \cup \{\tau\}$ we denote the set of all actions (with $\tau$ an internal action).

A communicating automaton %
is a finite set
of processes that exchange messages, each process being given as a finite LTS.
In the parametric setting,
the number of processes is not known beforehand. All processes
have the same state space, as in rendez-vous protocols~\cite{german1992reasoning}.
\begin{definition}
  \AP A ""Parametrized Communicating Automaton"" (\emph{PCA} for short) is
  a tuple
  $(\Stt,\transPCA,\initLoc)$.  Here $\Stt$ is the finite set of \emph{states},
  ${\transPCA} \subseteq \Stt\times \Act \times  \Stt$ the (finite) set of \emph{transitions}, and
  $\initLoc \in \Stt$ the~\emph{initial} state.
\end{definition}

We will use \emph{synchronous communication}  as synchronization mechanism. 
Given a "PCA" $\Aa=(\Stt,\transPCA,\initLoc)$ and a (tree) topology $\Tr=(P,\Ch)$,
we define a \emph{configuration} as a total function $\cnf:P\rightarrow \Stt$
mapping each process to its current state.
We denote by $\Conf_{\Aa,\Tr }$ the set of all configurations of $\Aa$ over $\Tr $. We call
$\cnfi{\Aa}{\Tr }$ the initial configuration of $\Aa$ over $\Tr $, with $\cnfi{\Aa}{\Tr }(p)=\initLoc$ for all $p\in P$ (we write $\cnf^\init$ when $\Aa$ and $\Tr$ are clear from the context).%

To define transitions we need to enrich actions by process names.
We denote by
$\Act_\Tr  = \{\sycom{p}{q}{m}\mid (p,q)\in \Ch,m\in\Msg\} \cup \{\sytau{p} \mid
p\in P\}$ the set of synchronous communication actions. As expected, 
$\sycom{p}{q}{m}$ is a synchronization between $p$ and $q$ with content $m$.

The transition relation $\rightarrow_{\Aa,\Tr }\mathop{\subseteq}\Conf_{\Aa,\Tr } \times \Act_\Tr  \times \Conf_{\Aa,\Tr }$ 
over the set $\Conf_{\Aa,\Tr }$ of configurations of $\Aa$ over $\Tr $ is defined as expected.
Given two configurations $\cnf,\cnf'\in\Conf_{\Aa,\Tr }$,
we let $\cnf\xrightarrow{a}_{\Aa,\Tr } \cnf'$ if either is true:
\begin{itemize}
  \item $a= \sycom{p}{q}{m}$ for some $(p,q)\in \Ch$ and $m\in \Msg$, $\cnf(p)\xrightarrow{\dw m}\cnf'(p)\in\transPCA$, $\cnf(q)\xrightarrow{\up m}\cnf'(q)\in\transPCA$,
  and $\cnf(t)=\cnf'(t)$ for all $t\in P\backslash\{p,q\}$.
  \item $a = \sytau{p}$ for some $p\in P$, $\cnf(p)\xrightarrow{\tau}\cnf'(p)\in\transPCA$ and $\cnf(q)=\cnf'(q)$ for all $q\neq p$.
\end{itemize}

A run of a "PCA" $\Aa$ over a topology $\Tr$ is a sequence $\rho=\cnf_0,a_1,\cnf_1,\dots,a_n,\cnf_n$ such that $\cnf_0=\cnfi{\Aa}{\Tr }$, and
$\cnf_i \xrightarrow{a_{i+1}}_{\Aa,\Tr }\cnf_{i+1}$ for all $0\leq i < n$. We call
$a_1\cdots a_n$ the
\emph{trace} of the run $\rho$, and denote it by $\trace(\rho)$.
A configuration $\cnf$ is \emph{reachable} if there is a run
ending in $\cnf$.

The projection of a trace $u\in{\Act_\Tr }^*$ over a process $p$, written as
$u|_p\in\Act^*$, is
the subsequence of  actions in which $p$ participates:
$\sycom{p}{q}{m}$ projects to $\dw m$, $\sycom{q}{p}{m}$ projects to $\up m$, and $\sytau{p}$ projects to $\tau$.

The size of a "PCA" $\Aa=(\Stt,\Act,\transPCA,\initLoc)$ is defined as $|S|+|\Delta|$.

A common task in parametrized verification is to check safety, \ie, that no error state is reachable, no matter how many processes participate in a run.
Formally,
the \textsc{Safety} problem asks,
given a "PCA" $\Aa$ and a state $s\in\Stt$,
whether
for all rooted trees $\Tr $, there is no
run ending in some configuration $\cnf$ such that  $\cnf(\root_\Tr)=s$.
For ease of reading, we will reason on the dual, \textsc{Coverability} problem:
\begin{center}
  \begin{tabularx}{\textwidth-2cm}{lX}
    \toprule
    \multicolumn{2}{c}{\textsc{Coverability}} \\
    \midrule
    \bfseries Input:
    & A "PCA" $\Aa = (\Stt,\transPCA,\initLoc)$ and a state $s\in\Stt$.
    \\
    \bfseries Output:
    & Yes if there exists some rooted tree $\Tr$ and a
    run of $\Aa$ over $\Tr$ ending in some configuration $\cnf$ such that  $\cnf(\root_\Tr)=s$.
    \\
    \bottomrule
  \end{tabularx}
\end{center}

\begin{remark}\label{rem:mutli_init_st}

  The reader may have noticed that all processes have the same underlying LTS,
  unlike our example in Figure~\ref{fig:example}. However, we can 
  simulate systems where the root has a different LTS, say
  $\text{LTS}_{root}$, by letting a process choose at the beginning of the run
  between $\text{LTS}_{root}$ and the LTS of the non-root processes.
  For this we create a new LTS as the disjoint union of  $\text{LTS}_{\root}$ (with
  initial state $s^\init_{\root}$) and the LTS of the non-root processes (with initial
  state $s_0$), adding a new initial state
  $s^\init$. State $s^\init$  has
  only two $\tau$-transitions, leading to $s^\init_{\root}$ and to $s_0$. We also remove all upward actions from $\Ss_{\root}$, and set the final state we want
  to reach $s_f$ in $\Ss_{\root}$. This way we can ensure that the root has
  chosen $\Ss_{\root}$ if $s_f$ is reached; if a non-root process chooses to go to
  $\Ss_{\root}$ then
  it cannot receive messages from its parent, so its subtree is disconnected.
  \lipicsEnd
\end{remark}

The Coverability problem defined above is undecidable, as to be expected
because processes on a branch of the tree 
are linearly ordered~\cite{AptK86}. A simple argument for undecidability in our setting is to reduce the halting problem for Minsky machines with two counters, 
using a branch to encode each counter in unary, together
with a special symbol at the end of the branch to do zero tests (see
also~\cite{emerson2003reasoning} for a similar construction). 

In order to regain
decidability we will study the behavior of "PCA" under two restrictions. The
first one, arguably the most natural, is to bound the depth of trees. The second one
is to bound
the number of alternations each process can do between synchronizing downwards
and upwards in the tree.

\AP A run $\rho$ is said to be ""$k$-phase-bounded"" if for every process $p \in P$,
$\trace(\rho)|_p\in ((\PSnd + \tau )^*  + (\PRec + \tau) ^*)^{\leq k}$.
In the example given in Figure~\ref{fig:example}, the tree depth  is bounded by
$d$, and the phase number for each
process is bounded by $3$: first receive an URL from the parent, then send URLs
and receive URL lists from the children, and finally send
an URL list back to the parent.
\AP A state $s\in\Stt$ is ""$(d,k)$-coverable"" in $\Aa$ if there exists a
tree $\Tr $ of depth at most $d$ and a "$k$-phase-bounded" run ending in some configuration $\cnf$ such that  $\cnf(\root_\Tr)=s$.
By writing $d=\infty$ we mean that there is no bound on the depth of
the topology. Likewise, for $k=\infty$ we mean that there is no bound on the number of
phases.

We will consider four variants of the problem, depending on whether $d$ and $k$ are fixed or are
part of the input.
If $d$ and/or $k$ are inputs of the problem,  they are in $\Nat$ and given in unary, otherwise they can be in
$\Nat\cup\{\infty\}$, with $\infty$ meaning unbounded.

\begin{center}
  \begin{minipage}{0.45\textwidth}    
    \begin{tabularx}{\textwidth-0.1cm}{l@{\,}X}
      \toprule
      \multicolumn{2}{c}{\textsc{$(d,k)$-Coverability}} \\
      \midrule
      \bfseries Input:
      & A "PCA" $\Aa$ and a state $s\in\Stt$.
      \\
      \bfseries Output:
      & Yes if $s$ is "$(d,k)$-coverable" in $\Aa$.
      \\
      \bottomrule
    \end{tabularx}
    \vspace{11.5pt}
  \end{minipage}
  \begin{minipage}{0.45\textwidth}
      \begin{tabularx}{\textwidth-0.1cm}{l@{\,}X}
        \toprule
        \multicolumn{2}{c}{\textsc{Coverability($d,k$)}} \\
        \midrule
        \bfseries Input:
        & A "PCA" $\Aa$, a state $s\in\Stt$ and two integers $d,k\in\Nat$.
        \\
        \bfseries Output:
        & Yes if $s$ is "$(d,k)$-coverable" in $\Aa$.
        \\
        \bottomrule
      \end{tabularx}
  \end{minipage}
\end{center}

\begin{center}
  \begin{minipage}{0.45\textwidth}
    \begin{tabularx}{\textwidth-0.1cm}{l@{\,}X}
      \toprule
      \multicolumn{2}{c}{\textsc{$(d)$-Coverability($k$)}} \\
      \midrule
      \bfseries Input:
      & A "PCA" $\Aa$, a state $s\in\Stt$ and an integer $k\in\Nat$.
      \\
      \bfseries Output:
      & Yes if $s$ is "$(d,k)$-coverable" in $\Aa$.
      \\
      \bottomrule
    \end{tabularx}
  \end{minipage}
    \begin{minipage}{0.45\textwidth}
    \begin{tabularx}{\textwidth-0.1cm}{l@{\,}X}
      \toprule
      \multicolumn{2}{c}{\textsc{$(k)$-Coverability($d$)}} \\
      \midrule
      \bfseries Input:
      & A "PCA" $\Aa$, a state $s\in\Stt$ and an integer $d\in\Nat$.
      \\
      \bfseries Output:
      & Yes if $s$ is "$(d,k)$-coverable" in $\Aa$.
      \\
      \bottomrule
    \end{tabularx}
  \end{minipage}
\end{center}

\begin{center}

\end{center}

The next two tables summarize our results:

\begin{figure}[htb]
  \begin{center}
    \renewcommand{\arraystretch}{1.3}
      \begin{tabular}{|c|c|c|c|c|}
        \hline        \diaghead(5,-4){cdd}%
      {$d$}{$k$}     & 1 & $\geq$2 & $\infty$\\
        \cline{1-4}
        1        &  \multicolumn{3}{l|}{{\EXPSPACE-c}}\\  
        \cline{1-1}\cline{3-4}
        \multirow[c]{2}*{$\geq 2$} &\multirow[c]{3}*{\makecell{(Lem~\ref{lem:exp-comp})\\\\{\color{white}\EXPSPACE-c}}} &  \multirow{3}*{\makecell{\EXPSPACE-c\\(Lem~\ref{lem:exp-comp}, Thm~\ref{thm:phase-bound-cov})}}& \multirow{2}*{\makecell{between $\FF_{\Omega_{d-1}}$\\and $\FF_{\Omega_{d}}$(Lem~\ref{lem:ncs-upper}, Lem~\ref{lem:ncs-lower-bound})}}\\
        &  & &  \\
        \cline{1-1}\cline{4-4}
        $\infty$   &  &  & Undec\\
        \cline{1-4}
      \end{tabular}
    \renewcommand{\arraystretch}{1}

    \caption{Complexity of \textsc{$(d,k)$-Coverability} for "PCA".}
  \label{fig:table_res}
  \end{center} 
\end{figure}
\begin{figure}[htb]
  \centering
  \begin{subfigure}{0.33\textwidth}
    \centering
    \begin{tabular}{|c|c|}
      \hline
        $d$ & \textsc{$(d)$-Coverability($k$)}\\
        \hline
        $1$ & \EXPSPACE-c (Lem~\ref{lem:exp-comp})\\
        \hline
        \makecell{$\geq 2$\\$\infty$} & \makecell{2\EXPSPACE-c\\(Thm~\ref{thm:2exp-hard}, Thm~\ref{thm:phase-bound-cov})}\\
        \hline
    \end{tabular}
  \end{subfigure}
  \begin{subfigure}{0.4\textwidth}
    \centering
    \begin{tabular}{|c|c|}
      \hline
        $k$ & \textsc{$(k)$-Coverability($d$)}\\
        \hline
        $\geq 1$ & \makecell{\EXPSPACE-c\\(Lem~\ref{lem:exp-comp}, Thm~\ref{thm:phase-bound-cov})}\\
        \hline
        $\infty$ & $\FF_{\varepsilon_0}$-c (Lem~\ref{lem:ncs-upper}, Lem~\ref{lem:ncs-lower-bound})\\
        \hline
    \end{tabular}
  \end{subfigure}
  \begin{subfigure}{0.25\textwidth}
    \centering
      \begin{tabular}{|c|}
        \hline
        \textsc{Coverability($d,k$)}\\ \hline
        \makecell{2\EXPSPACE-c\\(Thm~\ref{thm:2exp-hard}, Thm~\ref{thm:phase-bound-cov})}\\
        \hline
      \end{tabular}
  \end{subfigure}
  \caption{Complexity of \textsc{Coverability} for "PCA" when $d$ and/or $k$ are part of the input.}
\end{figure}

Let us comment on some boundary cases for our parameters.
If we restrict to $1$-phase-bounded runs, the root process can talk to its children;
but those processes  will not be able anymore to communicate with their children. Thus a state is
"$(\infty,1)$-coverable" iff it is "$(1,1)$-coverable". 
Similarly, if the depth is one then
every run is $1$-phase-bounded,
so a state is "$(1,\infty)$-coverable" iff it is "$(1,1)$-coverable".

\section{Depth-bounded "PCA"}
In this section, we will show that coverability of "PCA" on bounded-depth
trees is closely related  to 
coverability in "Nested Counter Systems"~\cite{decker2016freeze,balasubramanian2022complexity}.
We start by showing that $(1,k=\infty)$-\textsc{coverability} is equivalent to "VASS" \textsc{coverability}.

\subsection{Depth one  "PCA"}\label{sec:one}

Let $m$ be a positive integer.
\AP An $m$-""VASS"" is a tuple $\Vv=(Q,q^\init,\transVASS)$ where $Q$ is a finite set of states, $q^\init$ is the initial state, and
$\transVASS\subseteq Q\times\Int^m\times Q$ a transition relation.

 A \emph{configuration} of $\Vv$ is
a tuple $(q,\ivec)$ where $q\in Q$ is a state and $\ivec\in\Nat^m$ is a vector
of non-negative integers. The initial configuration is $(q^\init,\vzero)$.
The step relation between two configurations is defined by $(q,\ivec)\xrightarrow{\ivtr}(q',\ivec')$ if
$(q,\ivtr,q')\in\transVASS$ and $\ivec' = \ivec + \ivtr \geq \vzero$.
It is understood that $\leq$ denotes the component-wise extension of the usual order $\leq$ on $\Nat$.
A run of $\Vv$ is a sequence $\sigma = (q_0,\ivec_0)\xrightarrow{\ivtr_1} (q_1,\ivec_1)\xrightarrow{\ivtr_2} \dots \xrightarrow{\ivtr_n}(q_n,\ivec_n)$,
where $(q_0,\ivec_0)$ is initial.

Let $(q,\ivec)$ and $(q',\ivec')$ be two configurations. We say that $(q',\ivec')$ covers $(q,\ivec)$ if $q=q'$ and $\ivec \leq \ivec'$,
and we denote this as $(q,\ivec) \preceq (q',\ivec')$.

The size of an $m$-"VASS" $\Vv$ is
$|\Vv| = |Q| + \sum_{(q,\ivtr,q')\in\transVASS} \|\ivtr\|_1$ where
$\|\ivtr\|_1 = \sum_{i=1}^m |\ivtr[i]|$.

We recall the coverability problem for "VASS":

\begin{center}
  \begin{tabularx}{\textwidth-2cm}{lX}
    \toprule
    \multicolumn{2}{c}{\textsc{VASS-Coverability}} \\
    \midrule
    \bfseries Input:
    & An $m$-"VASS" $\Vv$  and a state $q\in Q$.
    \\
    \bfseries Output:
    & Yes if  $(q,\vec{0})$ is coverable in $\Vv$.
    \\
    \bottomrule
  \end{tabularx}
\end{center}

This problem is known to be \EXPSPACE-complete~\cite{lipton76,rackoff1978covering}.

\medskip

To show that \Reachdk{1}{k=\infty} for "PCA" is in \EXPSPACE, we can use counter abstraction~\cite{german1992reasoning}.
We simulate a "PCA" $\Aa=(\Stt,\Act,\Delta,s^\init)$ on trees of depth 1 by
counting how many children of the root are in a given state.
Thus we construct in polynomial time an $m$-"VASS" $\Vv$ with $m=|\Stt|$, in which coverability for a
state $s\in\Stt$ can be checked in exponential space~\cite{rackoff1978covering}.

For the lower bound, we reduce \textsc{VASS-Coverability} to \Reachdk{1}{k=\infty}.
Let $\Vv$ be an $m$-"VASS". We construct a "PCA" $\Aa$ where process $\root$ keeps track of the state and simulates the transitions of $\Vv$, while the children 
will store the counter values.
Any child process will be either in state  $s^\init$ or in $s_i$ ($1 \leq i \leq m$), and the number of children in state $s_i$ will represent the value of the $i$-th counter.

To do so, we replace each transition $q\xrightarrow{\ivtr}q'$ of $\Vv$ by a series of transitions
that will move $\ivtr[i]$ children from state $s^\init$ to state $s_i$ if $\ivtr[i]\geq 0$,
and will move $-\ivtr[i]$ children from state $s_i$ to state $s^\init$ if $\ivtr[i]<0$, for
each $1 \leq i \leq m$.

From this we get that configuration $(q,\vec{0})$ is coverable in $\Vv$ iff there exists a tree
$\Tr$ of depth 1 and some configuration $\cnf\in\Cc_{\Aa,\Tr}$ reachable
in $\Aa$ over $\Tr$ where $\cnf(\root) = q$.

\begin{lemma}\label{lem:exp-comp}
    \Reachdk{1}{k=\infty} for "PCA" is \EXPSPACE-complete.
\end{lemma}

\subsection{Coverability for "Nested Counter Systems"}
\label{subsec:ncs}

We have seen that for depth one, "PCA" are equivalent to "VASS".
It seems natural to look at nested counters for "PCA" over trees of fixed depth.
We will show in the next subsection a tight correspondence between "PCA" over trees of depth at most $d$ and
$d$-"Nested Counter Systems"~\cite{decker2016freeze,balasubramanian2022complexity}.
The latter manipulate higher-order counters, in a similar manner to Nested Petri Nets~\cite{lomazova1999some}.
We first recall the definition of "Nested Counter Systems".

Let $d\in \Nat$ be an integer. 
\AP A $d$-""Nested Counter System"" ($d$-\emph{NCS}) is a tuple $\Nn = (Q,\Delta)$ where $Q$ is a finite set of \emph{states}
and $\Delta \subseteq \bigcup_{1 \leq i,j \leq d+1} (Q^i\times Q^j)$ is a finite set of \emph{rules}.
The set of \emph{configurations} $\Cc_\Nn$ of $\Nn$ is is defined as the set of all
labelled rooted trees of depth at most $d$, with labels
from the set $Q$.
Formally,
a configuration $\ncnf\in\Cc_\Nn$ is a tuple
$\ncnf=(V,E,\root,h)$ where $(V,E,\root)$ is a rooted tree of depth at most $d$,
and $h : V \rightarrow Q$ is a labeling function.

The transition relation $\mathop{\act{}} \subseteq \Cc_\Nn \times \Cc_\Nn$ on
configurations is defined as follows.
Let $r=((q_0,\dots,q_i),(q_0',\dots,q_j'))\in\Delta$ be a rule, with
$0 \leq i \leq j \leq d$.
We say that a configuration $\ncnf$ moves to some
configuration $\ncnf'$ using rule $r$
(written $\ncnf\xrightarrow{r}\ncnf'$) if there is a path $v_0,\dots,v_i$ in $\ncnf$ starting at the root labeled by $q_0,\dots,q_i$, and $\ncnf'$ is obtained from $\ncnf$ by
changing the label of each $v_k$ to $q_k'$ for $0 \leq k \leq i$, and for $i< k \leq j$ creating a new vertex $v_k$ as the child of $v_{k-1}$, labeled by $q_k'$.
Similarly, suppose $r=((q_0,\dots,q_i),(q_0',\dots,q_j'))\in\Delta$ is a rule, with $0 \leq j < i \leq d$.
Then $\ncnf\xrightarrow{r}\ncnf'$ if there is
a path $v_0,\dots,v_i$ in $\ncnf$ starting at the root and labeled by
$q_0,\dots,q_i$, and $\ncnf'$ is obtained from $\ncnf$ by changing the label of
each $v_k$ to $q_k'$ for $0 \leq k \leq j$
and removing the subtree rooted at $v_{j+1}$.

We write  $\ncnf\rightarrow\ncnf'$ if there exists a rule
$r\in\Delta$ such that
$\ncnf\xrightarrow{r}\ncnf'$. Moreover, we say that $\ncnf'$ is \emph{reachable}
from $\ncnf$ (written $\ncnf\act{*}\ncnf'$)  if
there exist $\ncnf_0,\dots,\ncnf_n$, $n \ge 0$, such that $\ncnf=\ncnf_0 \act{} \cdots\act{}
\ncnf_n=\ncnf'$. 
We say that $\ncnf'$ \emph{covers} $\ncnf$ (written $\ncnf\preceq\ncnf'$) if $\ncnf$
is an induced subtree of $\ncnf'$ with the same root.
The coverability problem for "NCS" is:

\begin{center}
  \begin{tabularx}{\textwidth-2cm}{lX}
    \toprule
    \multicolumn{2}{c}{\textsc{$d$-"NCS" Coverability}} \\
    \midrule
    \bfseries Input:
    & A $d$-"NCS" $\Nn$ and two configurations $\ncnf_i,\ncnf_f\in\Cc_\Nn$.
    \\
    \bfseries Output:
    & Yes if there exists $\ncnf$ such that $\ncnf_i\act{*}\ncnf$ and $\ncnf_f\preceq\ncnf$.
    \\
    \bottomrule
  \end{tabularx}
\end{center}

It is readily seen that $1$-"NCS" are equivalent to "VASS",
so \textsc{$1$-"NCS" Coverability} is \EXPSPACE-complete.
To state the complexity of \textsc{$d$-"NCS" Coverability} with $d \geq 2$, we need to introduce the \emph{fast-growing hierarchy} of complexity classes~\cite{schmitz2016complexity}.

We represent ordinal numbers using the \emph{Cantor Normal Form} (CNF) as $\alpha =
\omega^{\alpha_1} + \dots + \omega^{\alpha_n}$, where $\alpha_1 \geq \cdots \geq
\alpha_n$ are ordinals
also written in CNF.
We denote \emph{limit ordinals} by $\lambda$. These are ordinals such that $\alpha + 1 < \lambda$ for every $\alpha < \lambda$.

A \emph{fundamental sequence} for a limit ordinal $\lambda$ is a sequence of ordinals $(\lambda_n)_{n\in\Nat}$ with supremum $\lambda$.
We consider the same fundamental sequence as in~\cite{decker2016freeze,bala2026complexityncs},
which is defined by
\[
  (\alpha + \omega^{\beta+1})_n = \alpha + \underbrace{\omega^\beta + \dots + \omega^\beta}_n \qquad\qquad \text{and}\qquad\qquad (\alpha + \omega^{\lambda'})_n = \alpha + \omega^{\lambda'_n}
\]
for ordinals $\beta$ and limit ordinals $\lambda'$. 
For example, a fundamental sequence for $\omega^2$ is given by $(\omega \cdot n)_{n\in \Nat}$, a fundamental sequence for $\omega^\omega$ by
$(\omega^n)_{n\in\Nat}$, etc.

The \emph{fast-growing hierarchy} is the ordinal-indexed family of functions $F_\alpha : \Nat \rightarrow \Nat$
for $\alpha < \varepsilon_0$ inductively defined by
\[
  \hfil F_0(n) = n+1 \quad\quad F_{\alpha+1}(n) = \underbrace{F_\alpha( \cdots (F_\alpha}_{n \text{ times}}(n)) \cdots ) \quad \text{and} \quad F_{\lambda}(n) = F_{\lambda_n}(n)
  \ .
\]
The corresponding hierarchy $(\FF_\alpha)_{\alpha < \varepsilon_0}$ of \emph{fast-growing complexity classes} is defined in terms of the fast-growing functions $F_\alpha$,
see~\cite{schmitz2016complexity} for details.
We recall that $\FF_\omega$ corresponds to
the class of problems solvable in Ackermannian time, and
$\FF_{\omega^\omega}$ to the class of problems solvable in hyper-Ackermannian time.

To state what is known about the complexity of coverability for "NCS",
we let $\Omega_n$ denote the
tower of $\omega$'s of height $n-1$,
i.e.,
$\Omega_1 = \omega$ and $\Omega_n = \omega^{\Omega_{n-1}}$.
The ordinal $\varepsilon_0$ is the limit ordinal with fundamental sequence
$(\Omega_n)_{n\in\Nat}$.

\begin{theorem}[\cite{decker2016freeze,bala2026complexityncs}]
  The complexity of \textsc{$d$-"NCS" Coverability} is between
  $\FF_{\Omega_{d-1}}$ and $\FF_{\Omega_{d}}$ when $d \geq 2$ is fixed, and
  it is $\FF_{\varepsilon_0}$-complete when $d$ is part of the
  input.
\end{theorem}

  The complexity of \textsc{$d$-"NCS" Coverability} for a fixed $d \geq 2$
  was shown in~\cite{decker2016freeze} to be between $\FF_{\Omega_{d-1}}$ and $\FF_{\Omega_{2d}}$.
  However,
  as pointed out in~\cite{bala2026complexityncs},
  the proof of the lower bound is flawed.
  The flawed construction of~\cite{decker2016freeze} is fixed in~\cite{bala2026complexityncs},
  albeit
  for the extended model of "NRCS" that adds reset operations to "NCS".
  More precisely,
  it is shown in~\cite{bala2026complexityncs} that \textsc{$d$-"NRCS" Coverability} is $\FF_{\Omega_{d}}$-complete
  (so the upper bound is reduced to $\FF_{\Omega_{d}}$,
  despite the additional reset operations).

  The lower bound of~\cite{bala2026complexityncs} does not (at least directly) apply to \textsc{$d$-"NCS" Coverability}.
  We show in the next subsection that
  \textsc{$d$-"NRCS" Coverability} can be reduced to \textsc{$(d+1)$-"NCS" Coverability}.
  This provides a proof of the claimed lower bound of~\cite{decker2016freeze}.

\subsection{From Nested Reset Counter Systems to "Nested Counter Systems"}
\label{sec:NRCS-to-NCS}

This subsection presents a simulation of Nested Reset Counter Systems by plain Nested Counter Systems.
We first recall the definition of Nested Reset Counter Systems from~\cite{bala2026complexityncs}.

\smallskip

\AP A $d$-""Nested Reset Counter System"" (\emph{$d$-NRCS}) is a tuple $\Nn = (Q, \Delta_u, \Delta_r)$ where
$Q$ is a finite set of \emph{states},
$\Delta_u \subseteq \bigcup_{1 \leq i,j \leq d+1} (Q^i\times Q^j)$ is a finite set of \emph{update rules}, and
$\Delta_r \subseteq \bigcup_{1 \leq i \leq d} (Q^i \times Q \times Q^i)$ is a finite set of \emph{reset rules}.
For clarity,
an update rule $((q_0,\dots,q_i), (q'_0,\dots,q'_j)) \in \Delta_u$ is written $(q_0, \ldots, q_i) \xrightarrow{\upd} (q'_0, \ldots, q'_j)$ and
a reset rule $((q_0,\dots,q_i), p, (q'_0,\dots,q'_i)) \in \Delta_r$ is written $(q_0, \ldots, q_i) \xrightarrow{\rst{p}} (q'_0, \ldots, q'_i)$.
Notice that a $d$-"NCS" (as defined in the previous subsection) is a tuple $(Q, \Delta)$ such that $(Q, \Delta, \emptyset)$ is a $d$-"NRCS".

The set of \emph{configurations} $\Cc_\Nn$ of a $d$-"NRCS" $\Nn = (Q, \Delta_u, \Delta_r)$ is identical to the set of configurations of the $d$-"NCS" $(Q, \Delta_u)$.
Similarly,
the transition relation $\xrightarrow{r}$ of an update rule $r \in \Delta_u$ is identical to the transition relation $\xrightarrow{r}$ of the $d$-"NCS" $(Q, \Delta_u)$.
Given a reset rule $r = (q_0, \ldots, q_i) \xrightarrow{\rst{p}} (q'_0, \ldots, q'_i)$ with
$0 \leq i < d$,
we say that a configuration $\ncnf$ moves to some
configuration $\ncnf'$ using rule $r$
(written $\ncnf\xrightarrow{r}\ncnf'$) if there is a path $v_0,\dots,v_i$ in $\ncnf$ starting at the root labeled by $q_0,\dots,q_i$, and $\ncnf'$ is obtained from $\ncnf$ by
changing the label of each $v_k$ to $q_k'$ for $0 \leq k \leq i$, and removing every subtree rooted at a child $v$ of $v_i$ such that $v$ is labeled by $p$.
The coverability problem for "NRCS" is defined in the same way as for "NCS".

\begin{lemma}
  \label{lem:nrcs-to-ncs}
  For every fixed $d \geq 1$,
  the \textsc{$d$-"NRCS" Coverability} problem is reducible in polynomial time to the \textsc{$(d+1)$-"NCS" Coverability} problem.
\end{lemma}

The rest of this subsection presents our reduction.
We focus on reset rules $(q_0, \ldots, q_i) \xrightarrow{\rst{p}} (q'_0, \ldots, q'_i)$ that are applied to nodes of depth at least one (the depth of the root is zero),
i.e.,
such that $i \geq 1$.
Our recipe to simulate such a reset rule on a configuration $\ncnf$ may be summarized as follows:
\begin{enumerate}
\item
  pick a path $v_0,\dots,v_i$ in $\ncnf$ starting at the root labeled by $q_0,\dots,q_i$,
\item
  change the label of each $v_k$ to $q_k'$ for $0 \leq k \leq i$,
\item
  copy the subtree rooted at the node $v_i$ except for the children $v$ of $v_i$ that are labeled with $p$, and
\item
  delete the subtree rooted at $v_i$.
\end{enumerate}
It is understood that the root $v'_i$ of the copy performed in step (3) is a child of $v_{i-1}$ (i.e., $v'_i$ is a sibling of $v_i$).
As we cannot create a sibling of the root, this approach only works for $i \geq 1$.

\smallskip

The full version~\cite{decker2016freeze:arxiv} of the
paper~\cite{decker2016freeze} shows, albeit succinctly, how to  copy  a
subtree in lossy
manner.
The main idea is to copy the source subtree twice, in DFS fashion,
using a label $\mathtt{i}$ to mark the source node and two labels $\mathtt{o_1}$ and $\mathtt{o_2}$ to mark the copied nodes.
The source subtree is deleted along the way (when the DFS moves up).

Our simulation of reset rules (formally presented below) is inspired from the copy gadget of~\cite{decker2016freeze:arxiv}.
This simulation only uses update rules and is \emph{lossy} in the following sense:
\begin{itemize}
\item
  any configuration obtained after the gadget is $\preceq$-smaller than the configuration obtained after the reset rule, and
\item
  there is a run of the gadget that is not lossy, i.e., that provides the configuration obtained after the reset rule.
\end{itemize}

We want to stress that this simulation only works for reset rules $(q_0, \ldots, q_i) \xrightarrow{\rst{p}} (q'_0, \ldots, q'_i)$ with $i \geq 1$.
To simulate an arbitrary $d$-"NRCS" $\Nn = (Q, \Delta_u, \Delta_r)$ containing possibly reset rules $(q_0) \xrightarrow{\rst{p}} (q'_0)$ that apply to the root,
one may simply introduce an extra level,
i.e.,
consider the $(d+1)$-"NRCS" $\Nn' = (Q \cup \{x\}, \Delta'_u, \Delta'_r)$ where
each update rule $(q_0, \ldots, q_i) \xrightarrow{\upd} (q'_0, \ldots, q'_j)$ is replaced by $(x, q_0, \ldots, q_i) \xrightarrow{\upd} (x, q'_0, \ldots, q'_j)$ and
each reset rule $(q_0, \ldots, q_i) \xrightarrow{\rst{p}} (q'_0, \ldots, q'_i)$ is replaced by $(x, q_0, \ldots, q_i) \xrightarrow{\rst{p}} (x, q'_0, \ldots, q'_i)$.
The $(d+1)$-"NRCS" $\Nn'$ contains no reset rule on the root, so it can be transformed into an equivalent (with respect to coverability) $(d+1)$-"NCS".

\medskip

Let us now formally present our simulation of reset rules with update rules.
Consider a $d$-"NRCS" $\Nn = (Q, \Delta_u, \Delta_r)$ and a reset rule $r = (q_0, \ldots, q_i) \xrightarrow{\rst{p}} (q'_0, \ldots, q'_i)$ with $d > i \geq 1$.
We introduce new states $\mathtt{init}(r)$, $\mathtt{par}(r)$, $\mathtt{cp}(r)$, $\mathtt{cpd}(r, q)$, $\mathtt{cpu}(r)$, $\mathtt{src}(q)$ and $\mathtt{dst}(q)$, where $q$ ranges over $Q$.
The rules of the gadget simulating $r$ are given below.
In these rules,
\begin{itemize}
\item
  $\sigma$ stands for the sequence $\sigma = q'_1, \ldots, q'_{i-1}$,
\item
  $m$ stands for an arbitrary integer in $\{0, \ldots, d-i-1\}$, and
\item
  $x_1, \ldots, x_m, y, z$ stand for arbitrary states in $Q$ (note that the sequence $x_1, \ldots, x_m$ is possibly empty).
\end{itemize}
As $d$ is fixed,
the gadget simulating $r$ contains polynomially many rules.

\subparagraph{Initialization: case $i = 1$}
\[
\begin{array}{rcl}
(q_0, q_1) & \xrightarrow{\upd} & (\mathtt{init}(r), \mathtt{src}(q'_1))
\\
(\mathtt{init}(r)) & \xrightarrow{\upd} & (\mathtt{cp}(r), \mathtt{dst}(q'_1))
\end{array}
\]

\subparagraph{Initialization: case $i \geq 2$}
\[
\begin{array}{rcl}
(q_0, \ldots, q_i) & \xrightarrow{\upd} & (\mathtt{init}(r), q'_1, \ldots, q'_{i-2}, \mathtt{par}(r), \mathtt{src}(q'_i))
\\
(\mathtt{init}(r), q'_1, \ldots, q'_{i-2}, \mathtt{par}(r)) & \xrightarrow{\upd} & (\mathtt{cp}(r), q'_1, \ldots, q'_{i-2}, q'_{i-1}, \mathtt{dst}(q'_i))
\end{array}
\]

\subparagraph{Move down and copy:}
\[
\begin{array}{rcl}
(\mathtt{cp}(r), \sigma, x_1, \ldots, x_m, \mathtt{src}(y), z) & \xrightarrow{\upd} & (\mathtt{cpd}(r, z), \sigma, x_1, \ldots, x_m, y, \mathtt{src}(z))
\\
(\mathtt{cpd}(r, z), \sigma, x_1, \ldots, x_m, \mathtt{dst}(y)) & \xrightarrow{\upd} & (\mathtt{cp}(r), \sigma, x_1, \ldots, x_m, y, \mathtt{dst}(z))
\end{array}
\]

Recall that we want to copy the subtree rooted at the node $v_i$ except for the children $v$ of $v_i$ that are labeled with $p$.
To do so,
\textbf{the above two rules are \emph{not} included when $m = 0$ and $z = p$}.

\subparagraph{Move up and delete source:}
\[
\begin{array}{rcl}
(\mathtt{cp}(r), \sigma, x_1, \ldots, x_m, y, \mathtt{src}(z)) & \xrightarrow{\upd} & (\mathtt{cpu}(r), \sigma, x_1, \ldots, x_m, \mathtt{src}(y))
\\
(\mathtt{cpu}(r), \sigma, x_1, \ldots, x_m, y, \mathtt{dst}(z)) & \xrightarrow{\upd} & (\mathtt{cp}(r), \sigma, x_1, \ldots, x_m, \mathtt{dst}(y), z)
\end{array}
\]

\subparagraph{End:}
\[
\begin{array}{rcl}
(\mathtt{cp}(r), \sigma, \mathtt{src}(y)) & \xrightarrow{\upd} & (\mathtt{end}(r), \sigma)
\\
(\mathtt{end}(r), \sigma, \mathtt{dst}(y)) & \xrightarrow{\upd} & (q'_0, \sigma, y)
\end{array}
\]

Note that in the two rules above we could equivalently replace $y$ by $q'_i$.

\subsection{From "PCA" to "Nested Counter Systems" and back}

We show in the remainder of this section that \Reachdk{d}{k=\infty} for "PCA"
is (polynomially-) equivalent to \textsc{$d$-"NCS" Coverability}.
We start with the upper bound, by providing a reduction from coverability for "PCA" over trees of depth at most $d$ to coverability for "NCS"
of nesting depth at most $d$.

Let $\Aa=(\Stt_\Aa,\Act,\Delta_\Aa,s^\init)$ be a "PCA" and $d\in\Nat$ a positive integer. We construct a
$d$-"NCS" $\Nn=(Q_\Nn,\Delta_\Nn)$, with $Q_\Nn=\Stt_\Aa \cup \{\start\}$ and 
$\Delta_\Nn$ defined 
as follows.
For the initialization, the root is in state $\start$ and it generates the tree
with a set of rules spawning new levels
in state $s^\init$. The initialization is finished with the root going from $\start$ to $s^\init$:
For all $0 \le i< d$, $((\start,\underbrace{s^\init,\dots,s^\init}_{i\text{ times}}),(\start,\underbrace{s^\init,\dots,s^\init}_{i\text{ times}},s^\init))\in\Delta_\Nn$, and
$((\start),(s^\init))\in\Delta_\Nn$.

For each transition, we have one rule per tree level:
For all $q\act{\dw m}q'\in\Delta_\Aa$, $s\act{\up
m}s'\in\Delta_\Aa$, $0\leq i <d$, and $q_0,\dots,q_{i-1} \in \Stt_\Aa$, we let $((q_0,\dots,q_{i-1},q,s),(q_0,\dots,q_{i-1},q',s'))\in\Delta_\Nn$;
for all $q\act{\tau}q'\in\Delta_\Aa$, $0\leq i \leq d$, and $q_0,\dots,q_{i-1} \in \Stt_\Aa$, we have $((q_0,\dots,q_{i-1},q),(q_0,\dots,q_{i-1},q'))\in\Delta_\Nn$.

This "NCS" simulates $\Aa$ for every tree of depth at most $d$. Let $\Tr$ be a tree of depth at most $d$. For every configuration $\cnf$ of $\Aa$
over $\Tr$ we can construct a configuration $\ncnf$
of $\Nn$ over $\Tr$ such that %
$h(p)=\cnf(p)$.
We have that a configuration $\cnf$ is reachable in $\Aa$ on trees of depth at most $d$ iff there exists a configuration $C$ of $\Nn$ simulating $\cnf$ that is reachable from the configuration
$\start$. So a state $s$ is "$(d,k=\infty)$-coverable" in $\Aa$ iff the
configuration with the root in state $s$ is coverable from configuration $\start$ in $\Nn$.

\begin{lemma}\label{lem:ncs-upper}
  The complexity of
  \Reachdk{d}{k=\infty} for "PCA" (with $d \in \Nat$) is in $\FF_{\Omega_d}$, and the
  complexity of $(k=\infty)$-\textsc{Coverability}$(d)$  is in
  $\FF_{\varepsilon_0}$.
\end{lemma}

\medskip

For the lower bound we  show how to simulate a $d$-"NCS" with a "PCA" on trees of
depth at most $d$. Informally, we use each process to store one state of
$\Nn$, with
its children representing its subconfiguration. When doing a transition, we nondeterministically synchronize a branch of the tree, making sure every node
of the branch has the correct state, and then change each state along the branch.

Let $d\in\Nat$ be a positive integer, $\Nn=(Q_\Nn,\Delta_\Nn)$ a $d$-"NCS", and
$\ncnf^\init$ and $\ncnf^\final$  two configurations of $\Nn$.
For the reduction it is convenient to construct a "PCA" $\Aa$ with a  distinct LTS for the root (see Remark~\ref{rem:mutli_init_st}).
The set of states of $\Bb$ is $Q_\Nn$, plus sets of intermediary states used in
initialization, application of rules and termination. The transitions are
obtained as follows.
For every rule, we initialize the transition and choose a branch by synchronizing processes from the root,
checking at the same time the states on the path.
Then we resolve the transition from the process at the deepest level to the root.
If the rule removes part of the configuration, we put the corresponding processes in a state
$s^\dead$, in which they will not be able to synchronize anymore, cutting the subconfiguration from the rest of the tree.

For the initialization, we construct the desired configuration of $\Aa$ from
the leaves to
the root. Each leaf sends the vertex it represents to its parent, and the latter waits to receive every vertex of its children before
sending itself to its parent, and we repeat this up to $\root_\ncnf$.
To check if the final configuration is covered, we use a similar construction.

\ifthenelse{\boolean{arXivLongVersion}}{%
A formal definition of the transitions used in this reduction can be seen in Appendix~\ref{app:NCS_to_PCA}.
}{}

\smallskip

The relationship between $\Nn$ and $\Aa$ can be formalised by the following
notion of simulation. Let $\ncnf=(V,E,\root_\ncnf,h)$ be a configuration of $\Nn$, and $\cnf$
be a configuration of $\Aa$ over some tree $\Tt=(P,\Ch)$. We say that $\cnf$
simulates 
$\ncnf$ if all  conditions below hold:
\begin{itemize}
  \item $\ncnf$ is a subtree of $\Tt$ with the same root.
  \item For every $v\in V$, $\cnf(v) = h(v)$.
  \item For every $v \in P \backslash V$, $\cnf(v) = s^\init$,
$\cnf(v) = s^\dead$ or there exists a ancestor $v'$ of $v$ in  $P \backslash V$ such that $\cnf(v') = s^\dead$.
\end{itemize}
For $\ncnf^\init =(V,E,\root_\ncnf,h)$, and for every tree $\Tt$
that has $(V,E)$ as subtree, the first configuration $\cnf$ reachable by $\Aa$ over
$\Tt$ such that $\cnf(\root) \in Q_\Nn$ simulates $\ncnf^\init$. If $\Tt$ is not big enough, then $\root$ will never reach a state in $Q_\Nn$.
We also have that for every configuration $\ncnf$ of $\Nn$ reachable from $\ncnf^\init$, there
exists a tree $\Tt$ and a reachable configuration $\cnf$ of $\Aa$ over $\Tt$ that simulates
$\ncnf$. Conversely, for any tree $\Tt$, if a reachable configuration $\cnf$ of $\Aa$ over $\Tt$
simulates some configuration $\ncnf$ of $\Nn$, then this configuration is reachable from $\ncnf^\init$.
From this we get that $\ncnf^\final$ is coverable from $\ncnf^\init$ in $\Nn$ iff 
$s^\final$ is "$(d,k=\infty)$-coverable" in $\Aa$.

\begin{lemma}\label{lem:ncs-lower-bound}
 The problem $(d,k=\infty)$-\textsc{Coverability}  for "PCA" (with $d \in \Nat$) is $\FF_{\Omega_{d-1}}$-hard,
  and $(k=\infty)$-\textsc{Coverability}$(d)$ is $\FF_{\varepsilon_0}$-hard.
\end{lemma}

\section{Phase-bounded "PCA"}

In this section we present our  results on the complexity of
the coverability problem for phase-bounded "PCA".

\subsection{Lower bound for phase-bounded PCA}
We start with the lower bound for \Reachdk{\infty}{k}, so for $k$ phases and no
restriction on the depth.
We show a reduction
from the coverability problem for a succinct model of Petri nets, called
Transducer Defined Petri Nets~\cite{BaumannMTZ20}. 
We recall that a Petri Net is given as a tuple $N=(P,T,F,p_0)$ with $P$ a finite set of places,
$T$ a finite set of transitions, $F\subseteq (P\times T) \cup (T \times P)$ the flow relation, and 
$p_0\in P$ the initial place. A marking of $N$ is a function $\mrk\in P
\rightarrow \Nat$, and we say
that there are  $\mrk(p)$ tokens on place $p$. A transition $t\in T$ is firable in a marking $\mrk$ if $m(p) > 0$ for all $p$
such that $(p,t)\in F$.
It can then be fired, leading to the marking $\mrk'$ with $\mrk'(p) = \mrk(p) + F(t,p) - F(p,t)$ for all $p\in P$.
We also write $\mrk\xrightarrow{t}\mrk'$.
We say that a marking $\mrk$ is reachable in $N$ if there exists a sequence $\mrk_0\xrightarrow{t_1}\mrk_1\dots\mrk_{n-1}\xrightarrow{t_n}\mrk$ 
where $\mrk_0(p_0) = 1$ and $\mrk_0(p) = 0$ for all $p\neq p_0$.
A marking $\mrk$ is coverable in $N$ if there exists a marking $\mrk'$ such that $\mrk(p) \leq \mrk'(p)$ for all $p\in P$.
The question whether a marking is coverable is \EXPSPACE-complete~\cite{lipton76,rackoff1978covering},
whereas the reachability of a marking is Ackermann-complete~\cite{LerouxS19,Leroux21,CzerwinskiO21}.

\smallskip

We fix an alphabet $\S$.
An \emph{$n$-ary transducer} $\Mm=(Q,q_0,Q_f,\Delta)$ over alphabet $\Sigma$
consists of a set of states $Q$, an initial state $q_0\in Q$, a set of final states $Q_f\subseteq Q$ and a transition relation 
$\Delta \subseteq Q \times \Sigma^n \times Q$. We write $q \xrightarrow{a_1,\dots,a_n} q'$ for a transition between $q$ and $q'$. The size of $\Mm$ is 
$|\Mm|=|\Delta|$. An $n$-ary transducer $\Mm$ defines the relation $\Ll(\Mm)$ containing the $n$-tuples $(w_1,\dots,w_n)$ of
words over $\S$ 
for which there exists a sequence of transitions
$q_0 \xrightarrow{a_{1,1},\dots,a_{n,1}}q_1
\xrightarrow{a_{1,2},\dots,a_{n,2}}\dots \xrightarrow{a_{1,m},\dots,a_{n,m}}
q_m$ of $\Mm$ with $q_m\in Q_f$
and $a_{i,1}\dots a_{i,m} = w_i$ for all $i \in [1 \dots n]$.

\begin{definition}
    \AP A ""transducer-defined Petri Net"" (\emph{TdPN})
    $\Nn=(\ell,\Tmove,\Tjoin,\Tfork,\winit)$ consists of an integer $\ell
    \in\Nat$, 
    a binary transducer $\Tmove$ and two ternary transducers $\Tjoin,\Tfork$,
    all over input/output alphabet $\Sigma$, and an initial word $\winit \in
    \Sigma^\ell$. The associated  Petri Net
    $\text{PN}(\Nn)=(P,T,F,p_0)$ has $P = \Sigma^\ell$ as set of places and set $T$ of
    transitions defined as disjoint union of:
    \begin{itemize}
            \item $T_m = \{(w,w') \in \Sigma^\ell \times \Sigma^\ell \mid (w,w')\in \Ll(\Tmove)\}$
            \item $T_j = \{(w,w',w'') \in \Sigma^\ell \times \Sigma^\ell\times \Sigma^\ell \mid (w,w',w'')\in \Ll(\Tjoin)\}$ and
            \item $T_f = \{(w,w',w'') \in \Sigma^\ell \times \Sigma^\ell\times \Sigma^\ell \mid (w,w',w'')\in \Ll(\Tfork)\}$
    \end{itemize}
 The initial place of $\text{PN}(\Nn)$ is $p_0 = w_\init$. The semantics of transitions in terms of
 flow relation is:
        \begin{itemize}
            \item if $ t = (p,p') \in T_m$ then $(p,t),(t,p')\in F$
            \item if $ t = (p,p',p'') \in T_j$ then $(p,t),(p',t),(t,p'')\in F$
            \item if $ t = (p,p',p'') \in T_f$ then $(p,t),(t,p'),(t,p'')\in F$
        \end{itemize}

    An accepting run of any of the transducers corresponds to a single
    transition of $\text{PN}(\Nn)$, and is
    called a transducer move. The size of $\Nn$ is defined as $|\Nn| = \ell +
    |\Tmove| + |\Tfork| + |\Tjoin|$, with $\ell$ in unary encoding.
\end{definition}

A marking $\mrk$ of a "TdPN" $\Nn$ is a total function $\mrk\in \Sigma^\ell
\rightarrow \Nat$. Coverability of markings (or places) is defined as usual.

\begin{center}
  \begin{tabularx}{\textwidth-2cm}{lX}
    \toprule
    \multicolumn{2}{c}{\textsc{TdPN-Coverability}} \\
    \midrule
    \bfseries Input:
    & A "TdPN" $\Nn$  and a word $\wfinal \in \Sigma^\ell$.
    \\
    \bfseries Output:
    & Yes if the place $\wfinal$ is coverable in $\text{PN}(\Nn)$.
    \\
    \bottomrule
  \end{tabularx}
\end{center}

We will use the following result:
\begin{theorem}[\cite{BaumannMTZ20}]
    \textsc{TdPN-Coverability} is 2\EXPSPACE-complete.
\end{theorem}

By abuse of language, we call  a total function $\mrk: \Sigma^l \rightarrow
\Nat$ a marking of $\Nn$,
and we say that there are $\mrk(w)$ tokens on place $w\in\Sigma^l$.

To reduce \textsc{TdPN-Coverability} to \textsc{$(d=\infty)$-Coverability($k$)}, we  construct a "PCA" on
a rooted tree of depth 2. Informally, the root simulates the transducers and each
child of the root stores a word, representing the encoding
of a place of the "TdPN". To do so, a node at depth
1 stores in each of its children a letter of the word, together with its
position from $[1 \dots \ell]$.

Each transition of the "TdPN" is initiated by the root by first synchronizing with
one, two or three children (depending on
the step), to define their role in the transition. Then $\root$ will communicate each letter of
the words needed for that transition to the participating children,
checking at the same time the validity of the transition of the chosen transducer.
At the
end of the transition, children that have participated
as inputs (outputs, resp.) go to state $\mathsf{dead}$
($\mathsf{full}$, resp.). The latter 
state means that such nodes are ready to
be used as inputs later.
At a given point, the number of children in state $\mathsf{full}$
with children corresponding to some word $w \in \S^\ell$ 
represents the number of tokens in the "TdPN" place named $w$.
\smallskip

In the reduction we  use a different LTS for $\root$. This is
not a restriction, see~Remark~\ref{rem:mutli_init_st}.

Let $\Nn=(\ell, \Tmove,\Tjoin,\Tfork,\winit)$ be the "TdPN"  and
$\wfinal\in\Sigma^\ell$ the final word.
We construct a "PCA" $\Aa(\Nn,\wfinal)$ with two distinct LTS,  $\Aa_\root$ for the root and $\Aa$ for the other processes, over the same set of actions $\Act$.
The state space of $\Aa_\root$ will be the disjoint union of 
$\{\init_r,\final_r,\idle\}$, a set of states for each of the two words
$\winit$ and
$\wfinal$, and one set of states per transducer. The general structure of this LTS is
shown in Figure~\ref{fig:root_auto_TdPN}.

\begin{figure}
    \begin{center}
        \begin{tikzpicture}[initial text=,
    every state/.style={inner sep=2pt,minimum size=10pt,ellipse}]
    \node[state,initial,circle] (init) at (0,0) {$\init_{r}$};
    \node[state,dashed, inner sep=10pt] (winit) at (3,0) {$\winit$};
    \node[state,circle] (final) at (0,-2) {$\final_{r}$};
    \node[state,dashed, inner sep=10pt] (wfinal) at (3,-2) {$\wfinal$};  
    \node[state,circle] (idle) at (6,-1) {$\idle$};
    \node[state, dashed, inner sep=10pt] (Tmove) at (6.75,0.5) {$\Tmove$};
    \node[state, dashed, inner sep=10pt] (Tfork) at (6.75,-2.5) {$\Tfork$};
    \node[state, dashed, inner sep=10pt] (Tjoin) at (8.25,-1) {$\Tjoin$};

    \path[->]   (init) edge[bend left=10pt,shorten >=-8pt] node[above] {$\dw\start^{1}$} (winit)
                (winit) edge[shorten <=-8pt,bend left = 10pt] (idle)
                (wfinal) edge[shorten <=-8pt,bend left = 10pt] (final)
                (idle) edge[shorten >=-8pt,bend left = 10pt,sloped] node[above] {$\dw\bar{\start}^{1}$} (wfinal)
                (idle) edge[shorten >=-8pt,out = 80, in = -120] (Tmove)
                (Tmove) edge[shorten <=-8pt,out=-90,in=50] (idle)
                (Tfork) edge[shorten <=-8pt,out=90,in=-50] (idle)
                (idle) edge[shorten >=-8pt,out = -80, in = +120] (Tfork)
                (idle) edge[shorten >=-8pt,bend left = 10pt] (Tjoin)
                (Tjoin) edge[shorten <=-8pt,bend left = 10pt] (idle)
                ;
\end{tikzpicture}
        \caption{Automaton for root. $\bar{\start}^{\text{arg}}$ means consuming
        a word, whereas $\start^{\text{arg}}$ means creating a word.}
        \label{fig:root_auto_TdPN}
    \end{center}
\end{figure}
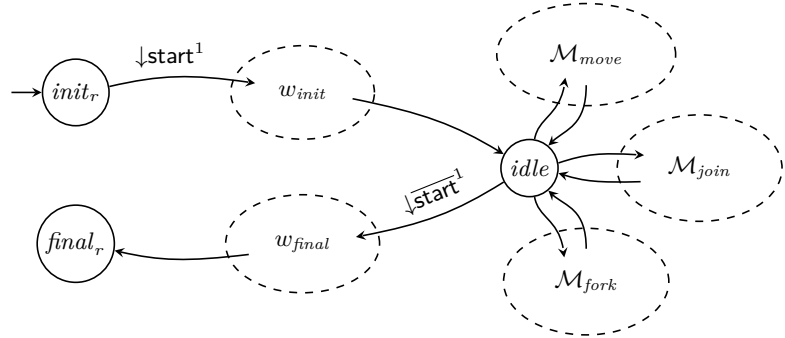

The LTS $\Aa$ will have two tasks. The first one is for  processes at level
1, and consists to store a word when asked by the root, and then
send it back later. Such a process $p$ is awakened when it synchronizes with
$\root$ with the message $\start^{\text{arg}}$ (for some $1 \le \text{arg}\leq3$). Process $p$ guesses
each letter of a word of length $\ell$, communicates it to one of its dormant
children together with its
position, and waits for the root to acknowledge the letter.
After the whole word has been stored, $p$ is in state $\full$.

Later, a level 1 process can be called by $\root$ to give its word back, when $\root$ synchronizes with it 
with some message $\bar{\start}^{\text{arg}}$. It will then
recover one by one
the letters of the word from its children and relay the information to the root.
After the word is sent, $p$ is in state $\dead$. The corresponding LTS
is depicted in Figure~\ref{fig:word_auto_TdPN}.
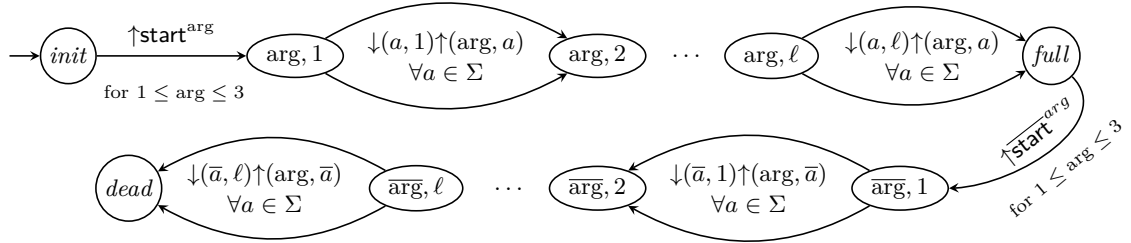
\begin{figure}
    \begin{center}
        \begin{tikzpicture}[initial text=,
    every state/.style={inner sep=2pt,minimum size=10pt,ellipse}]
    \newcommand{\dist}{2.5cm}
    \newcommand\secondline{-1.75}
    \node[state,initial,circle] (init) at (0,0) {\init};
    \node[state] (arg1) at (3,0) {$\text{arg},1$};
    \node   ()  at (1.4,-0.5) {\scriptsize for $1 \le\text{arg}\leq3$};
    \node[align=center] () at (5,0)  {$\dw(a,1)\up(\text{arg},a)$\\$\forall a \in\Sigma$};
    \node[state] (arg2) at (7,0)  {$\text{arg},2$};
    \node[align=center] () at (8.2,0)  {\dots};
    \node[state] (argl) at (9.3,0) {$\text{arg},\ell$};
    \node[align=center] () at (11.3,0)  {$\dw(a,\ell)\up(\text{arg},a)$\\$\forall a \in\Sigma$};
    \node[rotate=40]   ()  at (13.2,-1.5) {\scriptsize for $1 \le \text{arg}\leq3$};

    \node[state,circle] (full) at (13,0) {$\full$};
    \node[state] (narg1) at (11,\secondline) {$\bar{\text{arg}},1$};
    \node[align=center] () at (9,\secondline)  {$\dw(\bar{a},1)\up(\text{arg},\bar{a})$\\$\forall a \in\Sigma$};
    \node[state] (narg2) at (7,\secondline) {$\bar{\text{arg}},2$};
    \node       () at (5.85,\secondline)      {$\dots$};
    \node[state] (nargl) at (4.6,\secondline) {$\bar{\text{arg}},\ell$};
    \node[align=center] () at (2.6,\secondline)  {$\dw(\bar{a},\ell)\up(\text{arg},\bar{a})$\\$\forall a \in\Sigma$};

    \node[state,circle] (dead) at (0.8,\secondline) {$\dead$};

    \newcommand{\bendangle}{30}
    \path[->]   (init) edge node[above] {$\up\start^{\text{arg}}$} (arg1)
                (arg1) edge[bend left=\bendangle] (arg2)
                (arg1) edge[bend right=\bendangle] (arg2)
                (argl) edge[bend left=\bendangle] (full)
                (argl) edge[bend right=\bendangle] (full)
                (full) edge[in=0,out=-45] node[above, sloped] {$\up\bar{\start}^{arg}$} (narg1)
                (narg1) edge[bend left=\bendangle] (narg2)
                (narg1) edge[bend right=\bendangle] (narg2)
                (nargl) edge[bend left=\bendangle] (dead)
                (nargl) edge[bend right=\bendangle] (dead)
                ;
\end{tikzpicture}
        \caption{Automaton for processes at level 1.}
        \label{fig:word_auto_TdPN}
    \end{center}
\end{figure}
The second part of $\Aa$ will be for  processes at level 2. Such processes will only synchronize to store a letter and a position,
and wait to communicate back the corresponding letter and position. This LTS is depicted in Figure~\ref{fig:lettr_auto_TdPN}.
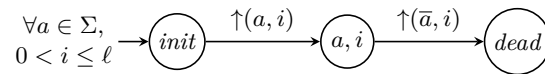
\begin{figure}
    \begin{center}
        \begin{tikzpicture}[initial text=,
    every state/.style={inner sep=2pt,minimum size=10pt,ellipse}]
    \newcommand{\dist}{2.5cm}
    \node[state,initial,circle] (init) at (0,0) {$\init$};
    \node[state,circle] (store) at (2.25,0) {$a,i$};
    \node[state,circle] (final) at (4.5,0) {$\dead$};
    \node[align=center]                      at (-1.5,0) {$\forall a\in\Sigma,$\\$0< i\leq \ell$};

    \path[->]   (init) edge[sloped] node[above] {$\up(a,i)$} (store)
                (store) edge[sloped] node[above] {$\up(\bar{a},i)$} (final)
                ;
\end{tikzpicture}
        \caption{Automaton for processes at level 2.}
        \label{fig:lettr_auto_TdPN}
    \end{center}
\end{figure}
Note that the states $\init$ and $\dead$ are common to both parts of this LTS.

We say that a configuration is \emph{stable} if the root is in state $\idle$ and all processes of depth 1 have their state in $\{\init,\full,\dead\}$.
These configurations represent markings of $\Nn$. We say that the stable configuration $\cnf$ simulates marking $\mrk$, noted $\cnf \sim \mrk$,
if for every word $w\in\Sigma^\ell$, there exist exactly $\mrk(w)$ processes at
level $1$ which have exactly one child in state $(w[i],i)$ for each $0< i \leq
\ell$,
and all their other children are in state $\init$. 
An example of a stable configuration simulating $\mrk=\llbracket ab, ab, ba
\rrbracket$ can be seen in Figure~\ref{fig:conf_from_mark}.
We give a small example of the execution of a move of a token from $ab$ to $ba$ in our system in Figure~\ref{fig:TDPN_move_figure}.
\begin{figure}[b]
    \begin{center}
        \begin{tikzpicture}[initial text=,
    every node/.style={draw,inner sep=2pt,minimum size=10pt,ellipse},
    every edge/.style={draw,dashed}]
    \node (root) at (0,0) {$\idle$};
    \node (w) at (-2,-1) {$\full$};
    \node[circle,inner sep = 1pt] (l1) at (-2.5,-2) {$a,1$};
    \node[circle,inner sep = 1pt] (l2) at (-1.5,-2) {$b,2$};
    \path[->] (root) edge (w) (w) edge (l1) (w) edge (l2);
    \node (w) at (0,-1) {$\full$};
    \node[circle,inner sep = 1pt] (l1) at (-0.5,-2) {$a,1$};
    \node[circle,inner sep = 1pt] (l2) at (0.5,-2) {$b,2$};
    \path[->] (root) edge (w) (w) edge (l1) (w) edge (l2);
    \node (w) at (2,-1) {$\full$};
    \node[circle,inner sep = 1pt] (l1) at (1.5,-2) {$b,1$};
    \node[circle,inner sep = 1pt] (l2) at (2.5,-2) {$a,2$};
    \path[->] (root) edge (w) (w) edge (l1) (w) edge (l2);

\end{tikzpicture}
        \caption{Configuration simulating $\mrk=\llbracket ab,ab,ba\rrbracket$. All unspecified
        processes have their state in $\{\init,\dead\}$.}
        \label{fig:conf_from_mark}
    \end{center}
\end{figure}
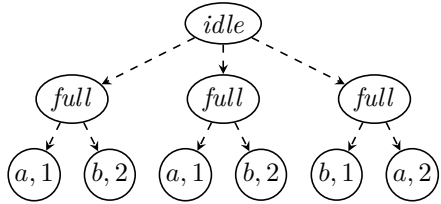

\begin{figure}[t]
    \begin{center}
        \begin{tikzpicture}[initial text=,
    every node/.style={inner sep=2pt,minimum size=10pt},rotate=-90]
    
    \newcommand\cholvl{-1.5}
    \newcommand\gchoolvl{-3}
    \newcommand\gchotlvl{-4.5}
    \newcommand\chtlvl{1.5}
    \newcommand\gchtolvl{3}
    \newcommand\gchttlvl{4.5}

    \newcommand\argcolor{black}
    \newcommand\poscolor{black}

    \tikzset{
        messagenodestyle/.style={fill=white},
        timelinestyle/.style={},
        statestyle/.style={draw,ellipse},
        treestyle/.style={dashed},
        messageedgestyle/.style={},
    }

    \node[statestyle] (root) at (0,0) {$\idle$};
    
    \node[statestyle] (ch1) at (1,\cholvl) {$\full$};
    \path[->,treestyle] (root) edge (ch1);
        \node[statestyle] (gch11) at (2,\gchoolvl) {$a,{\color{\poscolor}1}$};
        \node[statestyle] (gch12) at (2,\gchotlvl) {$b,{\color{\poscolor}2}$};
        \path[->,treestyle] (ch1) edge (gch11) (ch1) edge[bend right] (gch12);

    \node[statestyle] (ch2) at (1,\chtlvl) {$\init$};
    \path[->,treestyle] (root) edge (ch2);
        \node[statestyle] (gch21) at (2,\gchtolvl) {$\init$};
        \node[statestyle] (gch22) at (2,\gchttlvl) {$\init$};
        \path[->,treestyle] (ch2) edge (gch21) (ch2) edge[bend left] (gch22);

    \begin{scope}[xshift=-40pt]
        \node[statestyle] (rootend) at (8.25,0) {$\idle$};
        \path[timelinestyle] (root) edge (rootend);
        \path[timelinestyle,very thick] (5,0) edge (7.75,0);
        
        \node[statestyle] (ch1s1) at (3.75,\cholvl) {$\bar{{\color{\argcolor}1},{\color{\poscolor}1}}$};
        \node[statestyle] (ch1s2) at (5.75,\cholvl) {$\bar{{\color{\argcolor}1},{\color{\poscolor}2}}$};
        \node[statestyle] (ch1end) at (7.75,\cholvl) {$\dead$};
        \path[timelinestyle] (ch1) edge (ch1s1) (ch1s1) edge (ch1s2) (ch1s2) edge (ch1end);

        \node[statestyle] (ch2s1) at (4,\chtlvl) {${\color{\argcolor}2},{\color{\poscolor}1}$};
        \node[statestyle](ch2s2) at (6.25,\chtlvl) {${\color{\argcolor}2},{\color{\poscolor}2}$};
        \node[statestyle] (ch2end) at (8.25,\chtlvl) {$\full$};
        \path[timelinestyle] (ch2) edge (ch2s1) (ch2s1) edge (ch2s2) (ch2s2) edge (ch2end);

        \node[statestyle] (gch11end) at (5.5,\gchoolvl) {$\dead$};
        \path (gch11) edge (gch11end);
        \node[statestyle] (gch12end) at (7.25,\gchotlvl) {$\dead$};
        \path[timelinestyle] (gch12) edge (gch12end);

        \node[statestyle] (gch21end) at (5.5,\gchtolvl) {$b,{\color{\poscolor}1}$};
        \path[timelinestyle] (gch21) edge (gch21end);

        \node[statestyle] (gch22end) at (7.25,\gchttlvl) {$a,{\color{\poscolor}2}$};
        \path[timelinestyle] (gch22) edge (gch22end);

        \path[->,messageedgestyle]   (3,0) edge node[messagenodestyle] {$\bar{\start^{\color{\argcolor}1}}$} (3,\cholvl)
                    (3.25,0) edge node[messagenodestyle] {$\start^{\color{\argcolor}2}$} (3.25,\chtlvl)
                    (4.5,\cholvl) edge node[messagenodestyle] {$\bar{a},{\color{\poscolor}1}$} (4.5,\gchoolvl)
                    (5,0) edge node[messagenodestyle] {${\color{\argcolor}1},\bar{a}$} (5,\cholvl)
                    (4.75,\chtlvl) edge node[messagenodestyle] {$b,{\color{\poscolor}1}$} (4.75,\gchtolvl)
                    (5.5,0) edge node[messagenodestyle] {${\color{\argcolor}2},b$} (5.5,\chtlvl)
                    (6.75,\cholvl) edge node[messagenodestyle] {$\bar{b},{\color{\poscolor}2}$} (6.75,\gchotlvl)
                    (7.25,0) edge node[messagenodestyle] {${\color{\argcolor}1},\bar{b}$} (7.25,\cholvl)
                    (6.75,\chtlvl) edge node[messagenodestyle] {$a,{\color{\poscolor}2}$} (6.75,\gchttlvl)
                    (7.75,0) edge node[messagenodestyle]  {${\color{\argcolor}2},a$} (7.75,\chtlvl)
        ; 
    \end{scope}

\end{tikzpicture}
        \caption{Execution of a move transition from place $ab$ to place $ba$.
        The simulation of $\Tmove$ by root is depicted in bold.
        Dotted edges are tree edges, full edges represent transitions.}
        \label{fig:TDPN_move_figure}
    \end{center}
\end{figure}

We can show:

\begin{lemma}\label{lem:simulation}
    A marking $\mrk$ is reachable in $\Nn$ iff there exists a stable configuration $\cnf$ reachable in $\Aa(\Nn,\wfinal)$ such that
    $\cnf \sim \mrk$.
\end{lemma}

From Lemma~\ref{lem:simulation} we know that $\Aa(\Nn,\wfinal)$ can reach a
stable configuration with some process at level $1$ storing $\wfinal$ iff there exists
a reachable marking
with a token in $\wfinal$ in $\Nn$.
As $\root$ needs to consume $\wfinal$ as its final operation, we know that it can reach $\final_r$ iff $\wfinal$ is coverable in $\Nn$.

Moreover, the system $\Aa(\Nn,\wfinal)$ is phase-bounded. Indeed, $\root$ always does only one phase
and so do processes at level $2$. By construction,
processes
at level $1$ need $1+2\cdot\ell$ phases before reaching $\full$, and then  $1+2\cdot\ell$ other phases to go to state $\dead$. So they need at most $2+4\cdot\ell$ phases.

We can construct $\Aa(\Nn,\wfinal)$ in polynomial time, as for $\Aa_\root$ we
only need to copy and unwrap the transitions of each transducers into letter
transitions, and then
add the production of $\winit$ and consumption of $\wfinal$. For $\Aa$
the number of states and transitions is polynomial in the size of $\Sigma$ and $\ell$.
We can conclude from~\cite{BaumannMTZ20}:
\begin{theorem}\label{thm:2exp-hard}
    \textsc{$2$-Coverability($k$)} for "PCA" is 2\EXPSPACE-hard.
\end{theorem}

\subsection{Upper bound for phase-bounded "PCA"}
\label{subsec:upper-bound-ph-bounded}

In this section we show that \Reachdk{\infty}{k} is in \EXPSPACE\ and that \textsc{$(d=\infty)$-Coverability($k$)} is in $2\EXPSPACE$. To do so, we will prove that on phase-bounded runs,
the language of receives of each process is regular.
We will show how to compute automata for these languages on subtrees of depth one
first, and then apply the technique recursively.

\AP We define an ""Interface PCA"" (\emph{IPCA} for short) as a pair $\Aa=(\Aa_\root,\Aa_\ell)$, with the LTS $\Aa_\root$ for $\root$, and $\Aa_\ell$ for
non-root processes.
The only difference between "IPCA" and "PCA" is that the root process can do actions in $\PRec$ without a parent  to synchronize with.
So a run $\rho$ of a "IPCA" $\Aa$  is labelled by  $\trace(\rho)\in(\PRec\cup\Act_{\Tr})^*$ where 
$\projonproc{\trace(\rho)}{p}\in(\Act_{\Tr})^*$ for every $p \neq \root$. Recall
that $\Act_\Tr  = \{\sycom{p}{q}{m}\mid (p,q)\in C,m\in\Msg\} \cup \{\sytau{p} \mid
p\in P\}$ is the set of synchronous communication actions.

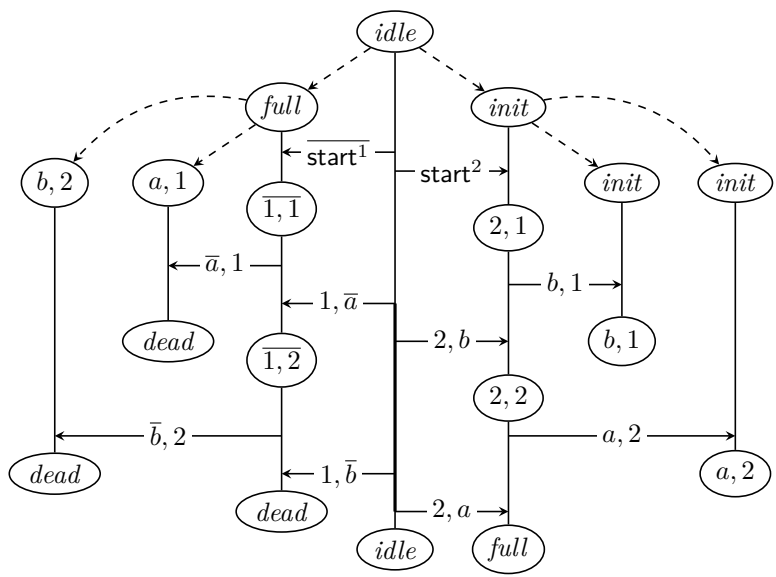
The language of an "IPCA" $\Aa$ is the set of words
$w\in \PRec^*$ such that there exists some tree $\Tr$ and  
a run $\rho$ of $\Aa$ over $\Tr$  with $w=\projonproc{\trace(\rho)}{\PRec}$. 
We denote 
by $L^{d \le 1}_k(\Aa)$ the
language of "IPCA" $\Aa$ over trees of depth $1$ and $k$-phase bounded
runs.

We already know that \Reachdk{1}{\infty} is equivalent to "VASS" coverability. We
will now extend this observation by showing that  languages of "IPCA" are also "VASS" languages.
To do so, we will enrich our "VASS" with action transitions.

\AP Let $m> 0$ be a strictly positive integer and $A$ be a finite alphabet.
An $m$-""VASS with actions"" from $A$ is a tuple $\Vv=(Q,q^\init,\VtransVASS,\Sigmat,\VtransPCA)$ where $Q$ is a finite set of states, $q^\init$ is the initial state, 
$\VtransVASS\subseteq Q\times \Int^m \times Q$ is a counter transition relation,
$\Sigmat =A \cup \{\tau\}$ is the alphabet and $\VtransPCA\subseteq Q \times
\Sigmat \times Q$ is an
action transition relation.

The size of an $m$-"VASS with actions" $\Vv$ is
$|\Vv| = |Q| + \sum_{(q,\ivtr,q')\in\VtransVASS} \|\ivtr\|_1 + |\VtransPCA|$.
Configurations are defined as for "VASS".
The step relation of a $m$-"VASS with actions" is
$(q,\ivec)\xrightarrow{a}(q',\ivec')$ if either 
$(q,a,q')\in\VtransVASS$, $a=\ivtr\in\Int^m$, $\ivec' = \ivec + \ivtr
\geq \vzero$, or $(q,a,q')\in\VtransPCA$, $a\in \Sigmat$, $\ivec = \ivec'
$.
In other words, actions (i.e., transitions from $\VtransPCA$) do not affect the counters.

A run of $\Vv$ is a sequence $\sigma = (q_0,\ivec_0)\xrightarrow{a_1} (q_1,\ivec_1)\xrightarrow{a_2} \dots \xrightarrow{a_n}(q_n,\ivec_n)$
where $(q_0,\ivec_0)$ is initial.
The trace of a run $\sigma$ is the sequence of actions labelling it, and is
denoted as $\trace(\sigma)\in {\Sigmat}^*$.
The language of a "VASS with actions" $\Vv$ is defined as $L(\Vv)=\{w \in A^*\mid w = \projonproc{\trace(\sigma)}{A}\text{ for some run }\sigma\}$.

We say that a run $\sigma$ of $\Vv$ is $k$-phase bounded if there exist $\sigma_1,\dots,\sigma_k$ such that $\sigma = \sigma_1 \cdot\sigma_2 \dots\sigma_k$, and
every $\sigma_i$ has all its transitions either in $\VtransVASS$ or in $\VtransPCA$.
We define $L_k(\Vv)$ as the language of $\Vv$ restricted to $k$-phase bounded runs.

We will now show how to construct a "VASS with actions" $\Vv$ from an "IPCA" $\Aa$ over trees of
depth $1$, which will simulate the "interface language" of $\Aa$.
\begin{restatable}{lemma}{lemIPCAtoVASS}\label{lem:IPCA-to-VASS}
    Let $\Aa=(\Aa_\root,\Aa_\ell)$ be an "IPCA" over trees of depth one.
    We can construct an $m$-"VASS with actions" $\Vv=(Q,q^\init,\VtransVASS,\Sigmat,\VtransPCA)$, with $Q = S_\root$, $m=|S_\ell|$, and $\Sigmat = \PRec \cup \{\tau\}$
     such that, for any $k\in\Natinf$,  $L_k(\Vv)= L^{d\leq1}_k(\Aa)$.
\end{restatable}

Next we show how to construct a finite automaton for the language of phase-bounded runs of a "VASS with actions". Let $m$ and $k$ be two integers.
Let $\Vv$ be an $m$-"VASS with actions", and $\sigma$ be a $k$-phase bounded 
run of $\Vv$. There exist $(q_1,\ivec_1),\dots,(q_k,\ivec_k)$ such that 
$\sigma = (q_0,\ivec_0),\sigma_1,(q_1,\ivec_1),\sigma_2,(q_2,\ivec_2)\dots(q_k,\ivec_k)$,
where every $\sigma_i$ is either only on $\VtransVASS$, or  only on $\VtransPCA$.
The important observation is that every subrun $\sigma_i$ over $\VtransPCA$ can be replaced by any
other subrun $\sigma_i'$ over $\VtransPCA$  starting  and ending in the same states as $\sigma_i$, as these phases have no effect on the counters.

\begin{lemma}\label{lem:pb-IVASS-to-auto}
    Let $m$ and $k$ be two integers. Let $\Vv=(Q,q^\init,\VtransVASS,A,\VtransPCA)$ be an $m$-"VASS with actions". There exists an automaton $\Bb$ with $O(k\times (2\times|Q|)^{k+1})$ states
    such that $L(\Bb)=L_k(\Vv)$, that can be constructed in $2^{O(m)}\cdot\log(k\times|\Vv|)$-space.
\end{lemma}
\begin{proof}{}
    \AP Let  $t=(q_0,d_1,q_1,d_2\dots,d_k,q_k)$ be a sequence of states and transition types,
    with $q_i\in Q$ and $d_i\in\{V,A\}$ for each $0<i\leq k$. We say that $t$ is
    \emph{valid} if there exists a run $\sigma =
    (q_0,\ivec_0),\sigma_1,(q_1,\ivec_1),\sigma_2,(q_2,\ivec_2),\dots,
    \sigma_k, (q_k,\ivec_k)$
    where for each $0<i\leq k$, subrun $\sigma_i$ is only over $\Delta_{d_i}$. In that case we say that $\sigma$ \emph{complies} with $t$.
    
    For every valid sequence $t$, we will construct a finite automaton $\Bb_t$ which language is the set of words $w\in A^*$
    such that there exists a run $\sigma = (q_0,\ivec_0),\sigma_1,(q_1,\ivec_1),\sigma_2,(q_2,\ivec_2)\dots(q_k,\ivec_k)$ complying with $t$, and with $w=\projonproc{\trace(\sigma)}{A}$.
    
    Let  $t=(q_0,d_1,q_1,d_2\dots,d_k,q_k)$ be a valid sequence.
    We construct $\Bb_t=(Q_{\Bb_t},q_{\Bb_t}^\init,\transPCA_{\Bb_t},q^\final_{\Bb_t})$, 
    with $Q_{\Bb_t}$ the set of states, $q_{\Bb_t}^\init\in Q_{\Bb_t}$ the initial state, 
    $\transPCA_{\Bb_t} \subseteq  Q_{\Bb_t} \times \Sigmat \times  Q_{\Bb_t}$
    the transition relation, and  
    $q^\final_{\Bb_t} \in Q_{\Bb_t}$ the final state, as follows:
    \begin{itemize}
        \item $Q_{\Bb_t} = Q\times k$, with $(q,i)\in Q_{\Bb_t}$ meaning that the automaton is in state $q$ and is simulating the $i$-th phase
        of a run.
        \item $q_{\Bb_t}^\init = (q_0,1)$.
        \item For each $0<i\leq k$, if $d_i = V$, then $(q_{i-1},i)\xrightarrow{\tau}(q_{i},i)\in\transPCA_{\Bb_t}$,
        otherwise if $d_i = A$, then for all $q\xrightarrow{a}q'\in\VtransPCA$, we have $(q,i)\xrightarrow{a}(q',i)\in\transPCA_{\Bb_t}$.
        Moreover, we have $(q_{i-1},i-1)\xrightarrow{\tau}(q_{i-1},i)\in\transPCA_{\Bb_t}$.
        \item $q^\final_{\Bb_t} =(q_k,k)$.
    \end{itemize}
    A run of $\Bb_t$ is a sequence $\rho = q_0 \xrightarrow{a_1} q_1 \dots q_{n-1} \xrightarrow{a_n} q_n$
    where $q_1 = q_{\Bb_t}^\init$, $q_n =q^\final_{\Bb_t} $, and for each $0< i \leq n$, $q_{i-1} \xrightarrow{a_i} q_{i} \in \transPCA_{\Bb_t}$.
    We denote by $\trace(\rho)\in A_\tau$ the trace of a run $\rho$.
    The language of $\Bb_t$ is the set $L(\Bb_t) = \{w \mid w = \projonproc{\trace(\rho)}{A} \text{ for some run }\rho\text{ of }\Bb_t\}$.
    Now we show that $L(\Bb_t)$ is exactly the set of all projections on $A$ of runs of $\Vv$ complying with $t$.
    Let $\sigma = (q_0,\ivec_0),\sigma_1,(q_1,\ivec_1),\sigma_2,(q_2,\ivec_2)\dots(q_k,\ivec_k)$ be a run complying with
    $t$, and let $w = \projonproc{\trace(\sigma)}{A}$. There exist $w_1,\dots,w_k$ such that $w_i = \projonproc{\trace(\sigma_i)}{A}$.
    We have that if $d_i = V$ then $w_i = \varepsilon$. Then when $d_i = A$, one can read $w_i$ between $(q_{i-1},i)$
    and $(q_i,i)$ in $\Bb_t$ by construction, and when $d_i = V$, one can go from state $(q_{i-1},i)$ to state $(q_i,i)$ in
    $\Bb_t$ reading $\tau$ (as when $d_i = V$ then $w_i = \varepsilon$). Moreover, one can change phases through transitions $(q_{i-1},i-1)\xrightarrow{\tau}(q_{i-1},i)$. So $w$ is accepted by $\Bb_t$.

    Now let $w\in L(\Bb_t)$. There exists  some run $\rho$ of $\Bb_t$ such that $\projonproc{\trace(\rho)}{A} = w$. By construction, there exist
    $\rho_1,\dots,\rho_k$ such that $\rho = \rho_1\cdot\rho_2\dots\rho_k$, where for each $0<i\leq k$ if $d_i = V$, then
    $\rho_i = (q_{i-1},i-1) \xrightarrow{\tau} (q_{i-1},i) \xrightarrow{\tau} (q_i,i)$, and if $d_i = A$ then $\rho_i = (q_{i-1},i-1)\xrightarrow{\tau}(q_{i-1},i)\xrightarrow{u_i} (q_i,i)$
    for some sequence of actions $u_i \in {\Sigmat}^*$.
    By construction, $u_i$ can be read in $\Vv$ between any two configurations $(q_{i-1},\ivec)$ and $(q_{i},\ivec)$.
    And as $t$ is valid, one can construct a $k$-bounded run $\sigma = (q_0,\ivec_0),\sigma_1,(q_1,\ivec_1),\sigma_2,(q_2,\ivec_2)\dots(q_k,\ivec_k)$
    where for each $0< i \leq k$, if $d_i = V$, then $\sigma_i$ makes only counter updates, and if $d_i=A$ then $\trace(\sigma_i)=u_i$.
    So $\sigma$ is complying with $t$ and $\projonproc{\trace(\sigma)}{A} = w$. 

    Finally we construct $\Bb$ by taking the disjoint union of every $\Bb_t$, over
    all valid sequences $t$.
    This automaton has at most $(k\times |Q|) \times (2\times |Q|)^{k}$ states.
   
    \medskip

    Finally we can check if a sequence $t = (q_1,d_1,q_2,\dots,d_{k-2},q_k)$ is
    valid through a reduction to a coverability question on an $m$-"VASS" with $|Q|\times k$
    states.
\ifthenelse{\boolean{arXivLongVersion}}{%
    The details are in Appendix~\ref{app:pb-IVASS-to-auto}.
}{}

     As coverability in an $n$-"VASS" $\Uu$ can be checked in $2^{O(n)}\times\log(|\Uu|)$-space (see~\cite{KunnemannMSSW25}), we can check if  $t = (q_1,d_1,q_2,\dots,q_k)$
     is valid in $2^{O(m)}\times\log(k\times|\Vv|)$-space.
\end{proof}
    
\medskip

Using the above reasoning we can check if a state of a "PCA" is
"$({d=\infty},{k})$-coverable", so for arbitrary depth. First we demonstrate how to use the technique
for depth $2$.
Let $k>0$ be a phase bound, $\Aa =(\Stt,\Act,\transPCA,\initLoc)$ a "PCA", and $s\in \Stt$ a state. By Lemma~\ref{lem:IPCA-to-VASS} we can simulate
the language of any process at depth $1$ on "$k$-phase-bounded" runs using an $m$-"VASS with actions" $\Vv=(\Stt,s^\init,\VtransVASS,\Sigmat,\VtransPCA)$,
where $m = |\Stt|$, $\Sigmat = \PRec \cup \{\tau\}$, that can be constructed in polynomial time. Moreover, using Lemma~\ref{lem:pb-IVASS-to-auto}, we can replace
$\Vv$ by a finite automaton $\Bb$ of size $O(k\times(2|S|)^{k+1})$ that can be built in $2^{O(|S|)}\times\log(O(k\times|S|\times|\transPCA|))$-space.
So we can reduce $(2,k)$-\textsc{coverability} to $(1,k)$-\textsc{coverability}.
Using the same technique as for the upper bound of Lemma~\ref{lem:exp-comp}, we can check if $s$ is coverable in $2^{O(k\times(2|S|)^{k+1})}\times\log(O(|S|))$-space.

If we want to apply the same technique for depth $3$, we first replace all processes of depth $2$ by  automata of size $O(k\times(2|S|)^{k+1})$, and then we
transform all processes of depth $1$ into  $m$-"VASS with actions" where $m = O(k\times(2|S|)^{k+1})$, and which statespace is still $\Stt$.
Using Lemma~\ref{lem:pb-IVASS-to-auto} again we can construct a finite automaton of size $O(k\times(2|S|)^{k+1})$ which can be built in
$2^{O(k\times(2|S|)^{k+1})}\times\log(O(k\times|S|\times|\transPCA|))$-space. Then one can check coverability of $s$ in $2^{O(k\times(2|S|)^{k+1})}\times\log(O(|S|))$-space again.

One can see that repeating the process of compression does not require more space after depth $3$, as the automaton produced will always be of size $O(k\times(2|S|)^{k+1})$.
Moreover, as the size of $\Bb$ is bounded, we know that if $s$ is coverable with $k$-bounded runs, it is coverable  with $k$-bounded runs over a tree of depth at most $O(k\times(2|S|)^{k+1})$.

We can  conclude:

\begin{theorem}\label{thm:phase-bound-cov}
    The problem  \textsc{$({d=\infty},{k})$-Coverability} for "PCA" (with $d \in \Nat$) is
    \EXPSPACE-complete. If $k$ is encoded in unary, then the problem \textsc{$(d=\infty)$-Coverability($k$)} is $2$\EXPSPACE-complete.
\end{theorem}

\section{Conclusion}

We investigated two kinds of bounds under which we
  determine the complexity  of safety checking of systems with rendez-vous
  synchronization over tree topologies. When
  bounding the depth we obtained that the complexity is related to the fast
  growing hierarchy. The second  bound on the phases  led to complexity
  \EXPSPACE-complete if the number of phases is fixed, resp.~2\EXPSPACE-complete if the number of phases is part of
 the input.

 An interesting question left open is whether given $d,k$ and a  "PCA"
 $\Aa$, if all executions of $\Aa$ over $d$-bounded trees
 are $k$ phase-bounded. A further natural question is to close the
 complexity gap when the depth is fixed.
 
\bibliography{biblio.bib}

\ifthenelse{\boolean{arXivLongVersion}}{%
\clearpage
\appendix

\section{Formal construction for the reduction of Coverability for "NCS" to Coverability for "PCA"}\label{app:NCS_to_PCA}

For every $r=((q_0,\dots,q_i),(q'_0,\dots,q'_j))\in\Delta_\Nn$, we initialize the transition and choose a branch by synchronizing processes from the root:
\begin{itemize}
  \item If $i = j = 0$, then $q_0\act{\tau}q_0'\in\Delta_\root$, otherwise $q_0\act{\dw(r,1)}(r,0)\in\Delta_\root$.
  \item For every $0<k<i$, $q_k\act{\up(r,k)}\act{\dw(r,k+1)}(r,k)\in\Delta_\Aa$.
  \item If $i\geq j$, then $q_i\act{\up(r,i)}(r,i)\in\Delta_\Aa$.
  \item If $i<j$, then  $q_i\act{\up(r,i)}\act{\dw(r,i+1)}(r,i)\in\Delta_\Aa$;
  for every $i<k< j$, $s^\init\act{\up(r,k)}\act{\dw(r,k+1)}(r,k)\in\Delta_\Aa$.
 Finally, $s^\init\act{\dw(r,j)}(r,j)\in\Delta_\Aa$.
\end{itemize}
Then we resolve the transition from the process at the deepest level to the root:
\begin{itemize}
  \item If $i>j$, then $(r,i)\act{\up(ok)}s^\dead\in\Delta_\Aa$;
  for every $j<k< i$ $(r,k)\act{\dw(ok)}\act{\up(ok)}s^\dead\in\Delta_\Aa$.
  Finally, if $j\neq 0$, then $(r,j)\act{\dw(ok)}\act{\up(ok)}q_j'\in\Delta_\Aa$.
  \item If $i \le j$, then  $(r,j)\act{\up(ok)}q_j'\in\Delta_\Aa$.
  \item In both cases, for every $0<k<j$, we have $(r,k)\act{\dw(ok)}\act{\up(ok)}q_k'\in\Delta_\Aa$, and $(r,0)\act{\dw(ok)}q_0'\in\Delta_\root$.
\end{itemize}
When a process reaches state  $s^\dead$, it will not be able to synchronize anymore, cutting its subconfiguration from the rest of the tree.

For the initialization, let $\ncnf^\init =(V,E,\root_\ncnf,h)$ be the initial
configuration of the $d$-"NCS". We first construct a corresponding configuration in $\Aa$ from
the leaves to
the root. The leaves will send the vertex they represent in $V$ to their parent, that will wait to receive every vertex of its children before
sending itself to its parent, and we repeat this up to $\root_\ncnf$:
\begin{itemize}
  \item For every leaf $v\in V$, $s^\init\act{\up(\init,v)}h(v)\in\Delta_\Aa$.
  \item For every other $v\neq \root_\ncnf$, let $v_1,\dots,v_n$ be the children of $v$, we have
  $s^\init\act{\dw(\init,v_1)}\dots\act{\dw(\init,v_n)}\act{\up(\init,v)}h(v)\in\Delta_\Aa$.
  \item For $\root_\ncnf$ let $v_1,\dots,v_n$ be the children of $\root_\ncnf$, we have
  $s^\init\act{\dw(\init,v_1)}\dots\act{\dw(\init,v_n)}h(\root_\ncnf)\in\Delta_\root$.
\end{itemize}
Once a process is chosen to be $\root_\ncnf$, it will only be allowed to do synchronizations with its
children. This way we ensure that for any tree $\Tt$, if any other process than $\root$ is chosen as $\root_\ncnf$, its
subconfiguration will be isolated from the rest of the tree (see Remark~\ref{rem:mutli_init_st}).

To check if the configuration $\ncnf^\final=(V,E,\root_\ncnf,h)$ is covered, we use a similar construction as initialization:
\begin{itemize}
  \item For every leaf $v\in V$, $h(v)\act{\up(\init,v)}s^\dead\in\Delta_\Aa$.
  \item For every other $v\neq \root_\ncnf$, let $v_1,\dots,v_n$ be the children of $v$, we have
  $h(v)\act{\dw(\init,v_1)}\dots\act{\dw(\init,v_n)}\act{\up(\init,v)}s^\dead\in\Delta_\Aa$.
  \item For $\root_\ncnf$ let $v_1,\dots,v_n$ be the children of $\root_\ncnf$, we have
  $h(\root_\ncnf)\act{\dw(\init,v_1)}\dots\act{\dw(\init,v_n)}s^\final\in\Delta_\root$.
\end{itemize}

\section{Proofs of Subsection~\ref{subsec:upper-bound-ph-bounded}}

\lemIPCAtoVASS*
\begin{proof}{}
Let $\Aa=(\Aa_\root, \Aa_\ell)$ be an "IPCA". We will construct an $m$-"VASS", for
$m=|S_\ell|$, by simulating up synchronizations of
$\root$ in $\Aa$ as $\VtransPCA$ transitions. As usual, we use the counters to store how many children of the root
are in each state, simulating down synchronizations from $\root$ as vector operations.
To avoid adding an extra phase at initialisation, we do not track the number of children processes in the initial state.
We define $\Vv=(Q,q^\init,\VtransVASS,\Sigmat,\VtransPCA)$
of dimension $m$ over alphabet $\Sigmat = \PRec \cup \{\tau\}$ as follows. We set $Q=S_\root$.
For $s\in\S_\ell$ we denote $\mathbf{v}_s \in \ZZ^{S_\ell}$ the vector with
$\mathbf{v}_s[s]=1$ and $\mathbf{v}_s[s']=0$  for $s' \not= s$.

The transition relation is defined as follows: 
\begin{itemize}
    \item For all $(s,\up m,s')\in\transPCA_\root$ we let $(s,\up m,s')\in\VtransPCA$.
    \item For all $(s_\root,\dw m,s_\root')\in\transPCA_\root$ and for all pairs
    of states $(s_\ell^\init,s_\ell')$ of $\Aa_\ell$ such that
    $s_\ell^\init\xrightarrow{\tau^*}\xrightarrow{\up
    m}\xrightarrow{\tau^*}s_\ell'$
    is a path in $\Aa_\ell$, we let $(s_\root,\ivtr,s_\root')\in\VtransVASS$,
    with $\ivtr = \mathbf{v}_{s_\ell'}$.
    \item For all $(s_\root,\dw m,s_\root')\in\transPCA_\root$ and for all pair of states $(s_\ell^\init,s_\ell')$ of $\Aa_\ell$ such that $s_\ell \neq s_\ell^\init$ and $s_\ell\xrightarrow{\tau^*}\xrightarrow{\up m}\xrightarrow{\tau^*}s_\ell'$
    is a path in $\Aa_\ell$, we let $(s_\root,\ivtr,s_\root')\in\VtransVASS$,
    with $\ivtr = \mathbf{v}_{s_\ell'} - \mathbf{v}_{s_\ell}$.
    \item For all $(s,\tau,s')\in\transPCA_\root$, we let $(s,\tau,s')\in\VtransPCA$ and $(s,\vzero,s')\in\VtransVASS$.
\end{itemize}

\medskip

Let $\Tr$ be a tree of depth one and $\cnf$ a configuration of $\Aa$ over $\Tr$.
We say that $(q,\ivec)$ simulates $\cnf$, denoted as $(q,\ivec)\sim\cnf$, if $\cnf(\root) = q$ and for every state $s\in\Ss_\ell\backslash\{s_\ell^\init\}$,
$|\cnf^{-1}(s)| = \ivec(s)$. 

We will now show that $L_k(\Vv)= L^{d\leq1}_k(\Aa)$.
First, let $w\in L^{d\leq1}_k(\Aa)$. Let $\Tr$ be a tree of depth $1$ such that there exists a $k$-bounded run $\rho = \rho_1\cdot\rho_2\dots\rho_k$ of
$\Aa$ over $\Tr$ with $w=\projonproc{\trace(\rho)}{\PRec}$. We suppose wlog.~that for every process $p\neq\root$ and for every $i \leq k$,
if $\projonproc{\trace(\rho_i)}{\PRec}\neq\varepsilon$, then $p$ does not act in $\rho_i$. Moreover,
if $p$ does not synchronize with $\root$, then it stays in state $s_\ell^\init$. If
this is not the case,
$\rho$ can be modified and still produce the same word $w$ without changing the number of phases.
We will construct a $k$-bounded run $\sigma$ of $\Vv$ such that $\projonproc{\trace(\sigma)}{\PRec}=w$.
Initially, every child process of the root is in state $s_\ell^\init$, so
$(s_\root^\init,\vzero) \sim \cnfi{\Aa}{\Tr}$.
We will show that for each $\rho_i$, we can construct a subrun $\sigma_i$ of $\Vv$ such that
if $\rho_i$ starts in configuration $\cnf$ and ends in configuration $\cnf'$, then $\sigma_i$
starts in configuration $(q,\ivec)\sim\cnf$ and ends in configuration $(q',\ivec')\sim\cnf'$.

First, suppose that $\projonproc{\trace(\rho_i)}{\PRec}\neq \varepsilon$, and $\rho_i$ starts in $\cnf$ and
ends in $\cnf'$. Let $(s,\ivec)$ be a configuration of $\Vv$ that simulates $\cnf$.
We know that for every action $a=\up m$ of $\rho_i$, the only moving process is $\root$.
Let $s$ be the state of $\root$ before $a$, and $s'$ the state of $\root$ after $a$.
We know there exists a transition $(s,\up m,s')\in\VtransPCA$, so we can replace every action of $\rho_i$ by an action
of $\Vv$, and construct a subrun $\sigma_i$ starting in $(s,\ivec)$ and ending in $(s',\ivec)\sim c'$. 

Now suppose that $\projonproc{\trace(\rho_i)}{\PRec} = \varepsilon$, and
$\rho_i$ starts in $\cnf$ and ends in $\cnf'$.
For each $p\neq\root$, we can suppose wlog.~that for some $n\in\Nat$ we have
$\rho_i = \rho_{1}'\cdot\rho_{2}'\dots\rho_{n}'$, where for all $j\leq n$ there
is some $p\neq\root$ with
$\projonproc{\trace(\rho_j')}{p} \in \tau^*\cdot\up m\cdot\tau^*$ for some
$m\in\Msg$, and for every $p' \neq p, \root$, 
$\projonproc{\trace(\rho_j')}{p'} = \varepsilon$.
If this is not the case, $\rho_i$ can be reordered and still reach $\cnf'$.
Now, let $\cnf_j'$ and $\cnf_{j+1}'$ be the configurations before and after $\rho_j'$.
We know that in $\rho_j'$ only $\root$ and one process $p$ moves. Moreover, $p$ will only
send one message in this subrun. So we can replace $\rho_j'$ by the transition $(s_\root,\ivtr,s'_\root)\in\VtransVASS$ with
$\cnf_j'(\root) = s$, $\cnf_{j+1}' = s'$, $\cnf_j'(p) = s_\ell$, $\cnf_{j+1}'(p)=s'_\ell$, and $\ivtr = \mathbf{v}_{s_\ell'} - \mathbf{v}_{s_\ell}$.

We can then construct a run $\sigma_i$ by replacing each $\rho_j'$ by one action, and if $\sigma_i$
starts in a configuration simulating $\cnf$, it will end in a configuration simulating $\cnf'$.

Finally, each $\sigma_i$ we constructed is exclusively over $\VtransPCA$ or over $\VtransVASS$, so the run
$\sigma = \sigma_1 \cdot \sigma_2 \dots \sigma_k$ is $k$-phase bounded.
\end{proof}

\subsection{Checking validity of a sequence in proof of Lemma~\ref{lem:pb-IVASS-to-auto}}\label{app:pb-IVASS-to-auto}

Let $t = (q_1,d_1,q_2,\dots,d_{k-2},q_k)$ be a sequence of states and transition
types. As explained below, 
checking if $t$ is valid is reducible to a coverability question on an $m$-"VASS" with $|Q|\times k$ states.

Let $\Uu_t=(Q_{\Uu_t},q^\init_{\Uu_t},\transPCA_{\Uu_t})$ be the following $m$-VASS:
\begin{itemize}
  \item $Q_{\Uu_t} = Q\times k$.
  \item $q^\init_{\Uu_t} = (q^\init,1)$
  \item For each $0<i\leq k$, if $d_i = V$, then for all $q\xrightarrow{\ivtr}q'\in\VtransVASS$, we have
  $(q,i)\xrightarrow{\ivtr}(q',i)\in\transPCA_{\Uu_t}$. Otherwise, if $d_i = A$, then for all $q\xrightarrow{a}q'\in\VtransPCA$,
  we have $(q,i)\xrightarrow{\vzero}(q',i)\in\transPCA_{\Uu_t}$. Moreover, we have $(q_{i-1},i-1)\xrightarrow{\vzero}(q_{i-1},i)\in\transPCA_{\Uu_t}$.
\end{itemize}

We show that the sequence $t$ is valid iff the configuration $((q_k,k),\vzero)$ is coverable in $\Uu_t$.
Suppose that $t$ is valid. Then there exists a run $\sigma^\Vv$ of $\Vv$ that complies with $t$. Let $\sigma^\Vv = (q_0,\ivec_0),\sigma^\Vv_1,(q_1,\ivec_1),\sigma^\Vv_2,(q_2,\ivec_2)\dots(q_k,\ivec_k)$
with every $\sigma^\Vv_i$ only on $\Delta_{d_i}$. We construct a run $\sigma_\Uu$ on $\Uu_t$ by doing the same counter updates done in $\sigma^\Vv_i$ when $d_i=V$,
and by replacing every action transitions by $\vzero$ transitions when $d_i=A$. Moreover, at the end of each $\sigma^\Vv_i$, we fire the transition $(q_{i-1},i-1)\xrightarrow{\vzero}(q_{i-1},i)$.
This run will reach a configuration $((q_k,k),\ivec)$, which covers $((q_k,k),\vzero)$.

Now suppose that we have a run $\sigma^\Uu$ that covers $((q_k,k),\vzero)$ in $\Uu_t$. There exist $\sigma^\Uu_1,\dots,\sigma^\Uu_k$ such that
$\sigma^\Uu = (q_0,0,\ivec_0) \xrightarrow{\sigma^\Uu_1} (q_1,0,\ivec_1) \xrightarrow{\vzero} (q_1,1,\ivec_1)
  \xrightarrow{\sigma^\Uu_1} \dots \xrightarrow{\sigma^\Uu_k} (q_k,k-1,\ivec_k) \xrightarrow{\vzero} (q_k,k,\ivec_k)$. By construction, for each $0<i\leq k$,
if $d_i = V$, one can do the same counter updates between states $(q,i)$ and
$(q',i)$ on $\Uu_t$ as between states $q$ and $q'$ of $\Vv$. So if
configuration
  $((q_i,i),\ivec_i)$ is reachable from $((q_{i-1},i),\ivec_{i-1})$ in $\Uu_t$, then configuration $(q_i,\ivec_i)$ is reachable from $(q_{i-1},\ivec_{i-1})$ in $\Vv$ doing the
  same counter updates.
  Now if $d_i = A$, then $\ivec_{i-1} = \ivec_i$, and by construction if $((q_i,i),\ivec_{i-1})$ is reachable from $((q_{i-1},i),\ivec_{i-1})$ in $\Uu_t$, then
$(q_i,\ivec_{i-1})$ is reachable from $(q_{i-1},\ivec_{i-1})$ doing only
action transitions. So we can construct a run of $\Vv$ of the form
  $\sigma^\Vv = (q_0,\ivec_0),\sigma^\Vv_1,(q_1,\ivec_1),\sigma^\Vv_2,(q_2,\ivec_2)\dots(q_k,\ivec_k)$ where every $\sigma_i$ is only on $\Delta_{d_i}$. Then $\sigma^\Vv$ complies
  with $t$.

}{}

\end{document}